\documentclass[12pt]{iopart}

\usepackage{iopams}

\expandafter\let\csname equation*\endcsname\relax

\expandafter\let\csname endequation*\endcsname\relax

\usepackage{amsmath}
\usepackage{amsthm}
\usepackage[utf8]{inputenc}
\usepackage{amssymb}
\usepackage{array}
\usepackage{mathtools}
\usepackage{enumitem}
\usepackage{accents}

\newcommand{\utilde}[1]{\underaccent{\tilde}{#1}}

\usepackage{hyperref}

\numberwithin{equation}{section}
\numberwithin{figure}{section}

\usepackage{tikz-cd}
\usetikzlibrary{babel}

\usepackage [english]{babel}
\usepackage [autostyle, english = american]{csquotes}
\MakeOuterQuote{"}

\usepackage[bottom]{footmisc}
\usepackage[justification=centering]{caption}

\usepackage{etoolbox}

\makeatletter
\newcommand{\mainmatter}{%
  \setcounter{footnote}{0}%
  \patchcmd{\@makefntext}{\fnsymbol}{\arabic}{}{}%
  \patchcmd{\@thefnmark}{\fnsymbol}{\arabic}{}{}%
  \def\@makefnmark{\textsuperscript{\arabic{footnote}}}%
}
\makeatother

\theoremstyle{plain}%
\newtheorem{theorem}{Theorem}[section]
\newtheorem{proposition}[theorem]{Proposition}%
\newtheorem{lemma}[theorem]{Lemma}

\theoremstyle{remark}%
\newtheorem{remark}[theorem]{Remark}%

\theoremstyle{definition}%
\newtheorem{definition}[theorem]{Definition}%

\begin{document}

\title[Massless single-particle bundles]{Bundle Theoretic Descriptions of Massless Single-Particle State Spaces; How do we perceive a moving quantum particle}

\author{Heon Lee}

\address{Department of Mathematical Sciences, Seoul National University, 1, Gwanak-ro, Gwanak-gu, Seoul, 08826, Republic of Korea}
\ead{heoney93@gmail.com}
\vspace{10pt}
\begin{indented}
\item[]February 2023
\end{indented}

\begin{abstract}
Recently, a bundle theoretic description of massive single-particle state spaces, which is better suited for Relativistic Quantum Information Theory than the ordinary Hilbert space description, has been suggested. However, the mathematical framework presented in that work does not apply to massless particles. It is because, unlike massive particles, massless particles cannot assume the zero momentum state and hence the mass shell associated with massless particles has non-trivial cohomology. To overcome this difficulty, this paper suggests a new framework that can be applied to massless particles. Applications to the cases of massless particles with spin-1 and 2, namely photon and graviton, will reveal that the field equations, the gauge conditions, and the gauge freedoms of Electromagnetism and General Relativity naturally arise as manifestations of an inertial observer's perception of the internal quantum states of a photon and a graviton, respectively. Finally, we show that gauge freedom is exhibited by all massless particles, except those with spin-0 and 1/2.
\end{abstract}

%
%
\submitto{\jpa}
%
%
%

\mainmatter

\section{Introduction}
\label{sec:1}

How do we perceive the internal quantum states of a moving particle? This issue of relativistic perception\footnote{For a precise definition of this terminology, see Sect.~2.5 of \cite{lee2022b}.} in Quantum Mechanics (QM) did not seem to have been taken seriously by the physics community before the dawn of Relativistic Quantum Information Theory (RQI).

In a generic communication scenario of RQI, one wants to use the internal quantum system of a single particle, which is usually finite-dimensional, as an information carrier. For example, in a scenario where an electron is exploited, the information sender encodes a qubit of information into the spin degrees of freedom of the electron, which is a 2-dimensional quantum Hilbert space, and sends it to a receiver, whose relativistic motion state may not be known to the sender. So, how the internal quantum state of the electron prepared by the sender is perceived by a different inertial observer becomes an essential question for RQI communication protocols.

However, a difficulty arises. Namely, a quantum particle cannot assume a definite momentum state. Hence, the best that the sender can do in the above scenario, for example, is to prepare a superposition of spin states of an electron corresponding to each possible momentum state that the electron can assume in such a way that the information that the sender wants to communicate is correctly perceived by an information receiver who might be in an arbitrary inertial motion state.

In this strategy, it is very important to take into account the fact that the internal quantum system of a moving particle might be perceived differently by a fixed inertial observer depending on each possible motion state of the particle. The most natural mathematical language which can accommodate this possibility is bundle theory. In this language, one can construct a vector bundle over the set of all possible motion states of a particle in whose fibers the moving particle's internal quantum states as perceived by a fixed inertial observer are encoded.

The recent paper \cite{lee2022b} considered two kinds of vector bundles to test whether one can construct a vector bundle over the set of all possible motion states of a \textit{massive} particle in such a way that the fibers of the bundle correctly reflects the perception of a fixed inertial observer. The first one, which was called \textit{boosting bundle}, was shown to be inadequate for this purpose even though the approaches based on this bundle have been the standard way to describe single-particle state spaces in the context of RQI. In contrast, the second one, which was called \textit{perception bundle}, was shown to fulfill this requirement. Moreover, using this bundle, \cite{lee2022b} showed that the Dirac equation and the Proca equations, which are fundamental equations of Quantum Field Theory (QFT) obeyed by massive particles with spin-1/2 and 1, respectively, emerge as manifestations of an inertial observer's perception of the internal quantum states of massive particles with spin-1/2 and 1, respectively.

However, one limitation of the work \cite{lee2022b} is that the theory applies only to massive particles. Massless particles need a separate treatment since for one thing, since massless particles cannot assume the zero momentum state, the bundles responsible for the description of massless particles are in general non-trivial, and for another, prominent massless particles such as photon and graviton exhibit \textit{parity inversion symmetry}, which must be taken into account in order to describe accurately these kinds of particles.

The present paper is for this treatment of massless particles. More precisely, this paper aims to develop a mathematical framework that will enable us to construct massless analogues of the perception bundles of \cite{lee2022b}, which should be vector bundles over the set of all possible motion states of massless particles, the fibers of which correctly reflect the perception of inertial observers. After the construction of the bundles, we are going to look at some of the theoretical implications that they entail. Specifically, by applying the massless perception bundle construction to the case of photon and graviton, which are examples of massless particles, we will see that Maxwell's equations in vacuum, the Lorentz gauge condition, and the gauge freedom of Electromagnetism (EM) emerge as manifestations of an inertial observer's perception of the internal quantum states of a photon, and Einstein's field equations in vacuum, the traceless-transverse gauge condition for symmetric 2-tensors, and the gauge freedom of linearized gravity in General Relativity (GR) emerge as manifestations of an inertial observer's perception of the internal quantum states of a graviton, respectively.

This paper is organized as follows. In Sect.~\ref{sec:2}, we identify all massless single-particle state spaces among the irreducible representations of the group $G:= \mathbb{R}^4 \ltimes SL(2,\mathbb{C})$ and show that they are classified by one numerical invariant called \textit{helicity}. Sect.~\ref{sec:3} explores some of the peculiarities of massless particles which prevent us from applying the same framework developed in \cite{lee2022b} for massive particles to the massless case, while deferring most of the proofs into Appendix. In particular, we will see that the first Chern classes of the bundles responsible for the description of massless particles are equal to the helicities of the corresponding particles and the direct sum of two massless particle representations with opposite helicities is an irreducible unitary representation of the group $\tilde{G}$, which is obtained by adjoining an involutive group element called \textit{parity inversion} to the group $G$. The latter fact is an essential ingredient for the definition of the \textit{spin} of a massless particle.

Sect.~\ref{sec:4} gives the boosting bundle description for massless particles, which is only a minor modification of that for massive particles. Based on this description, we present a brief survey on the early development of the RQI of massless particles. Sect.~\ref{sec:5} is about the perception bundle construction for massless particles. It gives rise to the concept of \textit{gauge freedom} and suggests a viewpoint on how one should interpret the wave functions representing massless particles. This construction and interpretation are applied to photon in Sect.~\ref{sec:6} and to graviton in Sect.~\ref{sec:7}, where one finds that the gauge freedom defined in Sect.~\ref{sec:5} becomes precisely those of EM and GR, respectively. The theoretical implications of the perception bundle description mentioned above are also investigated in these two sections. Finally, in Sect.~\ref{sec:8}, we apply the construction of Sect.~\ref{sec:5} to massless particles with arbitrary spin and show that the gauge freedom is exhibited by all massless particles except the ones with spin-0 and 1/2.

\paragraph{Notations and conventions}

\hfill

Before we begin, we summarize here some notations and conventions that will be used throughout the paper.

As the present paper is a continuation of the study \cite{lee2022b}, we will have many occasions to refer to the results of the latter. So, when referring to one of them, we will indicate it by prefixing "\cite{lee2022b}" to it. For example, "\cite{lee2022b}~Theorem~5.5" refers to Theorem~5.5 of \cite{lee2022b}.

Of course, we use all the notations adopted in \cite{lee2022b} so that for example, the Minkowski metric $\eta$ on $\mathbb{R}^4$ is denoted as
\begin{equation}\label{eq:1.1}
\langle x , y \rangle := \eta_{\mu \nu} x^\mu y^\nu = x_\mu y^\mu = x^0 y^0 - x^1 y^1 - x^2 y^2 - x^3 y^3
\end{equation}
and for $x \in \mathbb{R}^4$, we write
\begin{subequations}\label{eq:1.2}
\begin{eqnarray}
\tilde{x} &=& \begin{pmatrix} x^0 + x^3 & x^1 - i x^2 \\ x^1 + i x^2 & x^0 - x^3 \end{pmatrix}
 \\
\utilde{x} &=& \begin{pmatrix} x_0 + x_3 & x_1 - i x_2 \\ x_1 + i x_2 & x_0 - x_3  \end{pmatrix}
\end{eqnarray}
\end{subequations}
and the map $\kappa : SL(2, \mathbb{C}) \rightarrow SO^\uparrow (1,3)$ defined by
\begin{subequations}\label{eq:1.3}
\begin{eqnarray}
\left(\kappa (\Lambda) x \right)^\sim &=& \Lambda \tilde{x} \Lambda^\dagger  \\
\left( \kappa (\Lambda) x \right)_\sim &=& \Lambda^{\dagger -1} \utilde{x} \Lambda^{-1}
\end{eqnarray}
\end{subequations}
is a double covering homomorphism, and we write $\Lambda x := \kappa(\Lambda) x$ for $\Lambda \in SL(2, \mathbb{C})$ and $x \in \mathbb{R}^4$, etc. (For more details, see \cite{lee2022b}~Sect.~2.)

In addition to these, we write
\begin{equation}\label{eq:1.4}
| + \rangle = \begin{pmatrix} 1 \\ 0 \end{pmatrix}, \quad | - \rangle = \begin{pmatrix} 0 \\ 1 \end{pmatrix} \in \mathbb{C}^2
\end{equation}
to denote the helicity up/down states of a massless particle, respectively (cf. Sect.~\ref{sec:6.1}).

In contrast to Eq.~(\ref{eq:1.1}), the standard Hermitian dot product on $\mathbb{C}^n$ will be denoted as $\cdot$ for all $2 \leq n \in \mathbb{N}$, i.e., for $v,w \in \mathbb{C}^n$,
\begin{equation}\label{eq:1.5}
v \cdot w = \sum_{j=1} ^n \overline{v_j} w_j = v^\dagger w.
\end{equation}

\section{The single-particle state spaces for massless particles}\label{sec:2}

In \cite{lee2022b}~Sect.~3, we defined the \textit{single-particle state spaces} as the irreducible representation spaces of the group $G:= \mathbb{R}^4 \ltimes SL(2, \mathbb{C})$ (cf. \cite{wigner}), identified among them \textit{massive particles}, which are the irreducible representations of $G$ associated with the orbits $X_m ^\pm$ (cf. \cite{lee2022b}~Proposition~4.3), and observed that they are classified by two numerical invariants called \textit{mass} and \textit{spin}.

In this section, we continue the classification process that was set in \cite{lee2022b}~Remark~4.4 by analyzing the representations associated with the orbits $X_0 ^\pm$, which are called \textit{massless particles} (For a motivation for this terminology, see the remarks preceding \cite{lee2022b}~Definition~4.5). We will see that they are classified by one numerical invariant called \textit{helicity}.

\begin{proposition}\label{proposition:2.1}
Denote $H:= SL(2, \mathbb{C})$. Then, the little group for $p_0 ^{\pm}$ is
\begin{equation}\label{eq:2.1}
    K:=H_{p_0 ^{\pm}} = \left\{  \begin{pmatrix}z & b \\ 0 & \overline{z} \end{pmatrix} : z \in \mathbb{T}, b \in \mathbb{C} \right\}
\end{equation}
and the map $\Lambda : K \rightarrow E(2) := \mathbb{C} \ltimes U(1) $ defined by
\begin{equation}\label{eq:2.2}
\Lambda \begin{pmatrix} z & b \\ 0 & \overline{z} \end{pmatrix} = (z b , z^2)
\end{equation}
is a double covering homomorphism. Here, $\mathbb{T}$ is the unit circle.
\end{proposition}
\begin{proof}
The proof is a straightforward calculation using the definitions (cf. Eq.~(\ref{eq:1.3})). 
\end{proof}

Again, a version of Mackey machine can be used to find all irreducible representations of $K$. Those who are uninterested in the proof can jump right into Theorem~\ref{theorem:2.4}. We need the following two theorems which together are the summation of Theorems~6.37--6.41 of \cite{folland2015}.

\begin{theorem}\label{theorem:2.2}
Let $G$ be a Lie group. Suppose $N \trianglelefteq G$ is a closed abelian normal subgroup of $G$ and $G$ acts regularly on $\hat{N}$ (cf. \cite{lee2022b}~Sect.~4). Then, the followings hold.
\begin{enumerate}
\item If $\nu \in \hat{N}$ and $\sigma$ is an irreducible representation of $G_\nu$ such that $\sigma(n) = \nu (n) I_{\mathcal{H}_\sigma}$ for ${}^\forall n \in N$, then $\textup{Ind}_{G_\nu} ^G (\sigma )$ is an irreducible representation of $G$.

\item Every irreducible representation of $G$ is equivalent to one of this form.

\item $\textup{Ind}_{G_\nu} ^G (\nu \sigma) $ and $\textup{Ind}_{G_{\nu'}} ^G (\nu' \sigma') $ are equivalent if and only if $\nu$ and $\nu'$ belong to the same orbit, say $\nu' = x \nu$, and $h \mapsto \sigma(h)$ and $h \mapsto \sigma'(xhx^{-1})$ are equivalent representations of $H_\nu$.  
\end{enumerate}
\end{theorem}

\begin{theorem}\label{theorem:2.3}
Suppose $\nu \in \hat{N}$ can be extended to a unitary representation $\tilde{\nu} : G_\nu \rightarrow \mathbb{T}$ (note that $N \subseteq G_\nu$ by the definition of the $G$-action on $\hat{N}$). If $\rho : G_\nu /N \rightarrow U(\mathcal{H}_\rho)$ is an irreducible unitary representation, then the representation $\sigma : G_\nu \rightarrow U(\mathcal{H}_\rho)$ defined by
\begin{equation}\label{eq:2.3}
\sigma(y) = \tilde{\nu}(y) \rho (yN), \quad y \in G_\nu
\end{equation}
is an irreducible representation of $G_\nu$ satisfying $\sigma(n) = \nu (n) I_{\mathcal{H}_\rho}$. Furthermore, every irreducible representation of $G_\nu$ satisfying Eq.~(\ref{eq:2.3}) arises in this way. 
\end{theorem}

Let's try to apply these two theorems to the group $K$. Identify the closed abelian normal subgroup $M := \left\{\begin{pmatrix} 1 & c \\ 0 & 1 \end{pmatrix} : c \in \mathbb{C} \right\} \trianglelefteq K$. There is a topological group isomorphism $\mathbb{R}^2 \xrightarrow[\cong]{ } \hat{M}$ given by $b \mapsto \xi_b $ where $\xi_b : M \rightarrow \mathbb{T}$ is defined as $ \xi_b \left( \begin{pmatrix} 1 & c \\ 0 & 1 \end{pmatrix} \right):= e^{i b \cdot c}$ (here $\cdot$ denotes the Euclidean dot product in $\mathbb{R}^2$). The natural action of $K$ on $\hat{M} $ defined in \cite{lee2022b}~Sect.~4 is calculated as
\begin{align*}
\left[ \begin{pmatrix} z & a \\ 0 & \overline{z} \end{pmatrix} \xi_b \right]  \left( \begin{pmatrix} 1 &c \\ 0 & 1 \end{pmatrix} \right) = \xi_b \left(  \begin{pmatrix} z & a \\ 0 & \overline{z} \end{pmatrix}^{-1} \begin{pmatrix} 1 & c \\ 0 & 1 \end{pmatrix} \begin{pmatrix} z & a \\ 0 & \overline{z} \end{pmatrix} \right)
 = \xi_b \left( \begin{pmatrix} 1 & \overline{z}^2 c \\ 0 & 1 \end{pmatrix} \right) \\ = e^{ i b \cdot (\overline{z}^2 c)} =e^{i (z^2 b) \cdot c } = \xi_{z^2 b} \left( \begin{pmatrix} 1 & c \\ 0 & 1 \end{pmatrix} \right)
\end{align*}
for $z \in \mathbb{T} \cong SO(2) $ and $a,b, c \in \mathbb{C} \cong \mathbb{R}^2 $. So, under the identification $\mathbb{R}^2 \cong \hat{M}$, we see that the action of $K$ on $\hat{M} $ is translated into the action of $K$ on $\mathbb{R}^2$ given by
\begin{equation}\label{eq:2.4}
\begin{pmatrix} z & a \\ 0 & \overline{z} \end{pmatrix} \cdot b = z^2 b.
\end{equation} 

Therefore, for each $0 \neq b \in \mathbb{R}^2$, the isotropy group is given by
\begin{equation}\label{eq:2.5}
    K_b = \left\{ \begin{pmatrix} \pm 1 & a \\ 0 & \pm 1  \end{pmatrix} : a \in \mathbb{C} \right\}
\end{equation}
and of course $K_0 = K$. Note that for each $0 \neq b \in \mathbb{R}^2$, the map $\xi_b \in \hat{M}$ can be extended to $ K_b \rightarrow \mathbb{T}$ as
\begin{equation}\label{eq:2.6}
   \xi_b \left( \begin{pmatrix} \pm 1 & a \\ 0 & \pm 1 \end{pmatrix} \right) = e^{\pm i b \cdot a}
\end{equation}
and obviously, the trivial representation corresponding to $b = 0 $ extends to the trivial representation $K_0 = K \rightarrow \mathbb{T}$.

So, we can apply Theorems~\ref{theorem:2.2}--\ref{theorem:2.3} to the group $K$ to obtain all irreducible representations of $K$. Notice that $K_b/M \cong \mathbb{Z}/2 \mathbb{Z} $ for $ 0\neq b \in \mathbb{R}$.

\begin{theorem}\label{theorem:2.4}
There are two classes of irreducible representations of $K$. The first class arises from the trivial orbit $\{0\}$. They are the irreducible representations of the group \begin{equation*}
K_0 /M \cong \left\{ \begin{pmatrix} z & 0 \\ 0 & \overline{z} \end{pmatrix} : z \in \mathbb{T} \right \} \cong \mathbb{T}
\end{equation*}
lifted to $K$, i.e., the representations $\eta_{n/2}  : K \rightarrow U(1)$ given by
\begin{equation}\label{eq:2.7}
    \eta_{n/2} \begin{pmatrix} z & b \\ 0 & \overline{z} \end{pmatrix} = z^n \quad ( n \in \mathbb{Z}),
\end{equation}
which are distinct for each $n \in \mathbb{Z}$.

The second class arises from the orbits $K \cdot b = |b| \mathbb{T}$ for $0 \neq b \in \mathbb{R}^2$. They are the representations obtained by taking $\textup{Ind}_{K_b} ^K $ to the following representations of $K_b$
\begin{subequations}\label{eq:2.8}
\begin{align}
&\begin{pmatrix} \pm 1 & a \\ 0 & \pm 1 \end{pmatrix} \mapsto  \exp ( \pm i b \cdot a ) \\
&\begin{pmatrix} \pm 1 & a \\ 0 & \pm 1 \end{pmatrix} \mapsto \pm \exp ( \pm i b \cdot a ).
\end{align}
\end{subequations}

Representations arising from different circles are inequivalent.
\end{theorem}

In \cite{folland2008}, Folland states that the representations of the second class (that is, the ones in \eqref{eq:2.8}) do not correspond to any known physical particles. Weinberg explains this by saying "Massless particles are not observed to have any continuous degree of freedom (\cite{weinberg})" whereas the irreducible representations of the second class do exhibit a continuous degree of freedom by giving a continuum of representations.\footnote{They are called "continuous spin" in the physics literature.} So, from now on, we will exclude the second class, Eq.~(\ref{eq:2.8}), from our discussions.

Let's investigate the physical meaning of the constant $n \in \mathbb{Z}$ in Eq.~(\ref{eq:2.7}). Observe
\begin{equation}\label{eq:2.9}
\eta_{n/2} (e^{i \theta J^3} ) = \eta_{n/2} \begin{pmatrix} e^{ - \frac{i}{2} \theta} & 0 \\ 0 & e^{ \frac{i}{2} \theta} \end{pmatrix} = e^{- i \frac{n}{2} \theta}
\end{equation}
(cf. \cite{lee2022b}~Eq.~(4.19)) and hence
\begin{equation}\label{eq:2.10}
(\eta_{n/2} )_* (J^3) =  \left. \frac{d} {d \theta} \right|_{\theta = 0} (e^{- i \frac{ n}{2}  \theta} ) =  - i \frac{1}{2} n \in \mathfrak{gl}(\mathbb{C}) = \mathbb{C}.
\end{equation}

So, in particular, if we define $\hat{J}^3 := i (\eta_{n/2})_* (J^3)  $ as in \cite{lee2022b}~Eq.~(4.22), we obtain
\begin{equation}\label{eq:2.11}
\hat{J}^3 = \frac{1}{2} n \in \mathfrak{gl}(\mathbb{C}) = \mathbb{C}
\end{equation}
which shows that each vector in the irreducible representation space of $\eta_{n/2} $ (which is $\mathbb{C}$) is an eigenvector of the $x^3$-angular momentum operator $\hat{J}^3 \in \mathfrak{gl}(\mathbb{C})$ with eigenvalue $\frac{1}{2} n$ (this is why we indexed $\eta$ by $n/2$ rather than $n$).

Fix $ s \in \frac{1}{2} \mathbb{Z}$. Then,
\begin{equation}\label{eq:2.12}
\eta_s ( -I) = \eta_s (e^{2 \pi J^3} ) = e^{ 2 \pi (\eta_s)_* (J^3)} = e^{- 2 \pi i \hat{J}^3 } = (-1)^{2s}.
\end{equation}

Define
\begin{equation}\label{eq:2.13}
\pi_{0,  s} ^{\pm} := \pi_{p_0 ^{\pm} , \eta_{s}}
\end{equation}
for $s \in \frac{1}{2} \mathbb{Z}$ following the procedure of \cite{lee2022b}~Remark~4.4.

\begin{remark}\label{remark:2.5}
When a massless particle is moving through the $x^3$-axis (i.e., parallel to $p_0 ^\pm $), the angular momentum $x^3$-component is called the \textit{helicity} of the particle.
\end{remark}

This suggests the following definition.

\begin{definition}\label{definition:2.6}
The value $ s\in \frac{1}{2} \mathbb{Z}$ for the irreducible representation $\pi_{0, s} ^{\pm}$ is called the \textit{helicity} of the particles associated with this representation.
\end{definition}

Now, using \cite{lee2022b}~Theorem~4.2, we conclude

\begin{theorem}\label{theorem:2.7}
Excluding the unphysical class Eq.~(\ref{eq:2.8}), the irreducible representations associated with the orbits $X_0 ^{\pm}$ are classified by \textbf{helicity}. I.e., they are precisely given by
\begin{equation}\label{eq:2.14}
\left\{\pi_{0,s} ^{\pm} : s = 0 , \pm \frac{1}{2}, \pm 1, \pm \frac{3}{2}, \cdots  \right\}
\end{equation}
and they descend to projective representations of the group $\mathbb{R}^4 \ltimes SO^\uparrow (1,3)$ as in \cite{lee2022b}~Theorem~2.9 (cf. \cite{lee2022b}~Remark~2.11). In fact, $\pi_{0, s} ^\pm (-I) = (-1)^{2s} $.
\end{theorem}

\begin{proof}
The proof is exactly the same as the proof of \cite{lee2022b}~Theorem~4.10. The only difference is that we use Eq.~(\ref{eq:2.12}) to show the last assertion.
\end{proof}

\section{Some peculiarities of massless particles}\label{sec:3}

In this section, we discuss some of the peculiarities of massless particles which prevent us from applying the same framework developed in \cite{lee2022b} for massive particles to the massless case. The first one is that massless particles cannot assume the momentum zero state and hence the vector bundles responsible for the description of massless particles are not trivial in general. The second one concerns the subtle concept of \textit{parity inversion symmetry} which is exhibited by prominent massless particles such as photon and graviton.

\subsection{The Chern classes of the massless particle bundles}\label{sec:3.1}

As in \cite{lee2022b}, we want to construct the boosting and perception bundles for massless particles. However, the definitions of the two bundles as presented in \cite{lee2022b}~Appendix~A cannot be directly applied to the massless case. To see why, let's try to apply it to this case. First, we need to fix a $G$-invariant measure on $H/K \cong X_0 ^\pm$.

\begin{proposition}\label{proposition:3.1}
The following is a $G$-invariant measure on the orbit $X_0 ^\pm \cong \mathbb{R}^3 \setminus \{0\}$.
\begin{equation}\label{eq:3.1}
    d\mu (p) \cong \frac{ d^3 \mathbf{p}}{ |p^0| }
\end{equation}

Here the identification $\mathbb{R}^3 \setminus \{0\} \cong X_0 ^\pm$ is given by
\begin{equation}\label{eq:3.2}
\mathbf{p} \mapsto (\pm | \mathbf{p}|, \mathbf{p}).
\end{equation}
\end{proposition}
\begin{proof}
For a proof, see Ch.~1 of \cite{folland2008}.  
\end{proof}

Consider the principal $K$-bundle $H \rightarrow H/K$. Fix $s \in \frac{1}{2} \mathbb{Z}$ and consider the bundle $\mathcal{E}_s$ associated with the representation $\eta_s :H \rightarrow U(1)$ (cf. \cite{lee2022b}~Proposition~A.2). All of the discussions in \cite{lee2022b}~Appendix~A, except those regarding the perception and boosting bundles, can be applied to this setting. In particular, \cite{lee2022b}~Lemma~A.9 shows that $\mathcal{E}_{s}$ is an Hermitian $G$-bundle, called the \textit{primitive bundle associated with $\eta_s$}, with the metric
\begin{equation}\label{eq:3.3}
g\big( [B, \zeta_1] , [B, \zeta_2 ] ) = \overline{\zeta_1 } \zeta_2
\end{equation}
where $B \in H = SL( 2, \mathbb{C})$ and $\zeta_1, \zeta_2 \in \mathbb{C}$, and the $G$-action
\begin{equation}\label{eq:3.4}
\Lambda (a, A) [B , \zeta] = [AB , e^{-i \langle A B p_0 ^\pm , a \rangle } \zeta]
\end{equation}
where $A, B \in H$, $a \in \mathbb{R}^4$, and $\zeta \in \mathbb{C}$.

Also, \cite{lee2022b}~Lemma~A.10 yields the following theorem.

\begin{theorem}\label{theorem:3.2}
The irreducible representation $\pi_{0,s} ^\pm$, which represents massless particles with helicity-$s$, is equivalent to the induced representation $u : G \rightarrow U\big(L^2 (H/K, \mathcal{E}_{s}; \mu , g)\big)$ defined as, for $(a, A) \in G$ and $\psi \in L^2(H/K , \mathcal{E}_{s} ; \mu , g)$,
\begin{equation}\label{eq:3.5}
u (a,A) \psi = \Lambda(a,A) \circ \psi \circ l_A ^{-1}.
\end{equation} 
\end{theorem}

So, we see that the construction of the primitive bundle and expressing the induced representation in terms of the Hermitian $G$-bundle pose no problem. However, a problem arises when we try to define the perception and boosting bundles as in \cite{lee2022b}, which were the central focus of the paper.

To see what is at stake, observe that the definitions of the two bundles given there require the existence of a (continuous) representation of $H$ that extends $\eta_s$ and a (continuous) global section of the principal $K$-bundle $H \rightarrow H/K$, respectively. However, the following theorem implies that, for $s \neq 0$, both of them cannot exist since the existence of either of them would imply the triviality of the bundle (cf. \cite{lee2022b}~Theorem~A.7 and Lemma~A.7).

\begin{theorem}\label{theorem:3.3}
Let $\zeta$ be the 2-form on $X_0 ^\pm$ given by
\begin{equation}\label{eq:3.6}
\zeta := \frac{1}{4 \pi |p^0|}\left( p^1 d p^2 \wedge d p^3 + p^2 d p^3 \wedge d p^1 + p^3 d p^1 \wedge d p^2 \right),
\end{equation}
which is a generator of the cyclic group $ H^2 (X_0 ^\pm ; \mathbb{Z} ) \cong \mathbb{Z}$.

Then, for each $s \in \frac{1}{2} \mathbb{Z}$, the first Chern class of the bundle $\mathcal{E}_s$ is given by
\begin{equation}\label{eq:3.7}
c_1 ( \mathcal{E}_s ) = -2s [\zeta] \in H^2 (X_0 ^\pm ; \mathbb{Z}).
\end{equation}
\end{theorem}

Therefore, one cannot use the definitions of the two bundles as given in \cite{lee2022b}~Appendix~A to do a similar analysis for massless particles.

Apart from this remark that is relevant to the main body of the paper, the theorem also provides us with the following interesting fact: \textit{the helicity of a massless particle is a topological invariant, namely, the first Chern class, of the bundle underlying the representation describing the massless particle and it arises precisely because massless particles cannot have momentum zero.}

The proof of Theorem~\ref{theorem:3.3} is not simple and requires a fair amount of differential geometric machinery which will not be needed in the sequel. Hence it has been exiled to Appendix~\ref{sec:A}.

\subsection{Pairty inversion symmetry}\label{sec:3.2}

For the discussion of massless particles, we need to consider a subtle physical concept called \textit{parity}. This is because the paramount example of massless particle in the RQI literature, namely photon, exhibits a physical symmetry known as \textit{parity inversion symmetry} (cf. \cite{weinberg}). Therefore, we need to enlarge our symmetry group to include \textit{parity inversion}. (Cf. \cite{lee2022b}~Sect.~2)

For the moment, let $\tilde{G}$ be this enlarged symmetry group that is generated by $G$ and the parity inversion symmetry. By the discussion of \cite{lee2022b}~Sect.~2, we naturally make the following definitions.

\begin{definition}\label{definition:3.4}
A pair $(U, \mathcal{H})$ is called a \textit{quantum system with Lorentz and parity inversion symmetry} if $\mathcal{H}$ is a Hilbert space and $U: \tilde{G} \rightarrow U(\mathcal{H})$ is a unitary representation.

The irreducible unitary representation spaces of the group $\tilde{G}$ are called \textit{single-particle state spaces with parity inversion symmetry}. 
\end{definition}

These are the smallest possible quantum systems that are capable of testing the Lorentz and parity inversion symmetry in QM.

\begin{theorem}\label{theorem:3.5}
Let $ 0 \neq s \in \frac{1}{2} \mathbb{Z}$. Then, the representation $\pi_{0,s} ^\pm \oplus \pi_{0, -s} ^\pm$ of $G$ can be canonically extended to an irreducible unitary representation of $\tilde{G}$.
\end{theorem}

The extension of the representation $\pi_{0,s} ^+ \oplus \pi_{0, -s} ^+$ to $\tilde{G}$ for $0 \neq s \in \frac{1}{2} \mathbb{Z}$ has been used to describe a massless particle with parity inversion symmetry whose helicity takes values in $\{\pm s \}$. This motivates the following definition.

\begin{definition}\label{definition:3.6}
The nonnegative number $0 \leq s \in \frac{1}{2} \mathbb{Z}$ for the representation $\Pi_s := \pi_{0,s} ^\pm \oplus \pi_{0, -s} ^\pm$ is called the \textit{spin} of the massless particles with parity inversion symmetry that the canonical extension of $\Pi_s$ represents.
\end{definition}

The particles with parity inversion symmetry represented by the extension of $\Pi_s$ to $\tilde{G}$ will be called the \textit{massless particles with spin-$s$}. In this paper, we are only interested in the representation $\Pi_s$ and not in its extension to $\tilde{G}$. So, the proof of Theorem~\ref{theorem:3.5}, as well as the definition of the canonical extension of $\Pi_s$ to $\tilde{G}$, is deferred to Appendix~\ref{sec:B}.

\section{The boosting bundle description for massless particles; The RQI of massless particles}\label{sec:4}

We have just seen that the mathematical framework presented in \cite{lee2022b} does not apply to the massless particle case. However, at least for boosting bundle description, only a minor modification is required. In this section, we obtain the boosting bundle description for massless particles and briefly look at how it has been used in the RQI literature. The perception bundles will be taken up in the next section.

For the rest of the paper, we retain the notations of Sects.~\ref{sec:2}--\ref{sec:3} for the subgroups of $G := \mathbb{R}^4 \ltimes SL(2, \mathbb{C})$. In particular, $H = SL(2, \mathbb{C})$ and $K= H_{p_0 ^+}$ and hence $H/K \cong X_0 ^+$. We also fix a \textit{positive} half integer $ 0 <s \in \frac{1}{2} \mathbb{Z}$ and denote $\tilde{\eta}_{s} := \eta_s \oplus \eta_{-s}$ and $\Pi_{s} := \pi_{0,s} ^+ \oplus \pi_{0, -s} ^+$ (cf. Definition~\ref{definition:3.6}).

Moreover, we restrict our attention to the positive mass shell $X_0 := X_0 ^+$ and hence suppress all the $+$-superscripts from the expressions.\footnote{This is to follow the usual convention of the physics literature that deals with massless particles. At the time of writing, we are not aware of any treatment of the representations associated with the orbit $X_0 ^-$.} The Lorentz invariant measure $\mu$ and the identification $\mathbb{R}^3 \setminus \{0\} \cong X_0 $ as given in Proposition~\ref{proposition:3.1} will be used throughout.

\subsection{The boosting bundles for massless particles}\label{sec:4.1}

\paragraph{A choice of boostings}

\hfill

To obtain the boosting bundle description, we need to choose a global section for the bundle $H \rightarrow H/K \cong X_0 $ (cf. \cite{lee2022b}~Sect.~2.4). As we saw in Sect.~\ref{sec:3.1}, this bundle has no continuous global section. Hence, we have to content ourselves with a Borel one.

Fix $p \in X_0$, write $\hat {\mathbf{p}} := \frac{1}{|\mathbf{p}|} \mathbf{p}$. Since the boosting
\begin{equation*}
B_0 (|\mathbf{p}|) := e^{(\log |\mathbf{p}|) K^3} = \begin{pmatrix} \sqrt{|\mathbf{p}|} & 0 \\ 0 & \sqrt{|\mathbf{p}|}^{-1} \end{pmatrix}
\end{equation*}
maps $p_0 = (1,0,0,1)$ to $ (|\mathbf{p}|, 0 , 0 , |\mathbf{p}|)$, the element
\begin{equation}\label{eq:4.1}
L(p):= R(\hat{\mathbf{p}}) B_0 (|\mathbf{p}|) \in H
\end{equation}
maps $p_0 $ to $p$, where $R( \hat{\mathbf{p}}) \in SU(2)$ is any rotation that takes $\hat{\mathbf{z}}:= (0,0,1) \in \mathbb{R}^3$ to $\hat{\mathbf{p}} \in \mathbb{R}^3$. Irrespective of the exact form of $R(\hat{\mathbf{p}})$, this has been the standard choice of boostings for massless particles (cf. \cite{weinberg}).

Note that since $B_0 (|\mathbf{p}|)$ is a continuous function of $p \in X_0 $, one can never make a choice of the rotations $R(\hat{\mathbf{p}})$ in such a way that it becomes a continuous function of $p \in X_0 $ because that would imply the triviality of the bundle $H \rightarrow H/K$, which contradicts the result of Sect.~\ref{sec:3.1}.

Even though the exact formula for $R(\hat{\mathbf{p}})$ is not relevant to our discussions, we make the following choice for definiteness. For $ p \in X_0 $, we can write $\hat{\mathbf{p}} = (\sin \theta \cos \phi, \sin \theta \sin \phi, \cos \theta) \in \mathbb{R}^3$ for some $0 \leq \theta \leq \pi, \hspace{0.1cm} 0 \leq \phi \leq 2 \pi$. Observe that the matrix
\begin{align}\label{eq:4.2}
R(\hat{\mathbf{p}}) := e^{\phi J^3} e^{\theta J^2} e^{-\phi J^3} = \begin{pmatrix} e^{-i \frac{\phi}{2} } & 0 \\ 0 & e^{i \frac{\phi}{2}} \end{pmatrix} \begin{pmatrix} \cos \frac{\theta}{2} & -\sin \frac{\theta}{2} \\ \sin \frac{\theta}{2} & \cos \frac{\theta}{2} \end{pmatrix} \begin{pmatrix} e^{i \frac{\phi}{2} } & 0 \\ 0 & e^{-i \frac{\phi}{2}} \end{pmatrix} \nonumber \\
= \frac{1}{2 \cos \frac{\theta}{2}} \begin{pmatrix} 1+ \cos \theta & - e^{-i \phi} \sin \theta \\ e^{i \phi} \sin \theta & 1 + \cos \theta \end{pmatrix} \nonumber \\
= \frac{1}{\sqrt{2 (1+ \hat{\mathbf{p}}^3 ) } } \begin{pmatrix} 1+ \hat{\mathbf{p}}^3 & -( \hat{\mathbf{p}}^1 - i \hat{\mathbf{p}}^2 ) \\ \hat{\mathbf{p}}^1 + i \hat{\mathbf{p}}^2 & 1+ \hat{\mathbf{p}}^3 \end{pmatrix}
\end{align}
is a rotation that takes $\hat{\mathbf{z}} := (0,0,1) \in \mathbb{R}^3$ to $\hat{\mathbf{p}}$. Note that this map is a continuous map of $p \in X_0$ except on the connected ray in $X_0 \cong \mathbb{R}^3 \setminus \{0\}$ spanned by $-\hat{\mathbf{z}}$, which is a $\mu$-null set. Hence, Eq.~(\ref{eq:4.1}) is Borel measurable.

However, one should note that following \cite{weinberg}, the majority of RQI papers (e.g., the influential paper \cite{peres2003}) used $R(\hat{\mathbf{p}}) = e^{\phi J^3} e^{\theta J^2}$ for the rotations, unlike our choice Eq.~(\ref{eq:4.2}). Ours has been adopted, for example, in \cite{caban2003}.

\paragraph{The bundles}

\hfill

Since the induced representation construction respects the direct sum (easy to check from \cite{lee2022b}~Definition~4.1), we see
\begin{equation}\label{eq:4.3}
\Pi_s = \pi_{0,s} \oplus \pi_{0, -s} \cong \textup{Ind}_{G_{p_0}} ^G (e^{-i \langle p_0 , \cdot \rangle} [\eta_s \oplus \eta_{-s}]) = \textup{Ind}_{G_{p_0}} ^G (e^{-i \langle p_0 , \cdot \rangle} \tilde{\eta}_s).
\end{equation}

Applying \cite{lee2022b}~Theorem~5.5 to the representation $\tilde{\eta}_s : K \rightarrow U(2)$ and the (Borel) global section $L$ defined in Eqs.~(\ref{eq:4.1})--(\ref{eq:4.2}), we obtain the following theorem.

\begin{theorem}\label{theorem:4.1}
For $\Lambda \in H$ and $p \in X_0$, define the \textbf{Wigner transformation} (cf. \cite{lee2022b}~Eq.~(5.7)) as
\begin{equation}\label{eq:4.4}
W_L ( \Lambda, p) := L(\Lambda p)^{-1} \Lambda L(p) \in K.
\end{equation}

Then, the trivial bundle $X_0 \times \mathbb{C}^2$ endowed with the structure listed in Table~\ref{tab:1} becomes a Borel Hermitian $G$-bundle\footnote{Borel Hermitian $G$-bundles are the objects satisfying the same conditions as Hermitian $G$-bundles (cf. \cite{lee2022b}~Definition~5.2), but, all the appearing maps (such as local trivializations and the $G$-action, etc) are Borel measurable (but, not necessarily continuous).} over $X_0$, whose associated induced representation $U_{L,s}$ is unitarily equivalent to $\Pi_s$.
\end{theorem}

\begin{table}[h]

\caption{The boosting bundle description for massless particles}
\label{tab:1}

\centering
\begin{tabular}{|m{1.3cm}|m{8.4cm}|}
\hline\noalign{\smallskip}
  & $E_{L, s}$  (The boosting bundle) \\
\noalign{\smallskip}\hline\noalign{\smallskip}
Bundle  &  $X_0 \times \mathbb{C}^2$ \\
\noalign{\smallskip}\hline\noalign{\smallskip}
Metric &  $h_L \Big( (p , v) , (p ,w) \Big)= v \cdot w$ \\
\noalign{\smallskip}\hline\noalign{\smallskip}
Action &   $\lambda_{L,s} (a,\Lambda) (p , v) = \Big(\Lambda p , e^{-i \langle \Lambda p , a \rangle} \tilde{\eta}_s \big( W_L(\Lambda , p)\big) v \Big)$ \\
\noalign{\smallskip}\hline\noalign{\smallskip}
Space & $\mathcal{H}_{L,s} = L^2 ( X_0  ; \mu ) \otimes \mathbb{C}^2$ \\
\noalign{\smallskip}\hline\noalign{\smallskip}
$\Pi_s $ & $U_{L,s} (a,\Lambda )\psi = \lambda_{L,s} (a,\Lambda) \circ \psi \circ \Lambda^{-1}$ \\
\noalign{\smallskip}\hline
\end{tabular}

\end{table}

\begin{proof}
Straightforward (cf. \cite{lee2022b}~Appendix~A). Note that all the involved maps are Borel measurable and hence the representation theoretic information is retained when passing from the primitive bundle for $\tilde{\eta}_s$ to the bundle of Table~\ref{tab:1}.
\end{proof}

We will denote the Borel Hermitian bundle of Table~\ref{tab:1} as $E_{L, s}$ and call it the \textit{boosting bundle associated with $L$}.

\paragraph{The representations}

\hfill

Because $W(\Lambda, p) \in K$ for $ { }^\forall \Lambda \in H$ and ${ }^\forall p \in X_0$, we see that there exists $\Theta (\Lambda , p ) \in \mathbb{R}$, which is called the \textit{Wigner phase}, such that
\begin{equation}\label{eq:4.5}
\tilde{\eta}_s (W( \Lambda, p ) ) = \begin{pmatrix} e^{is \Theta (\Lambda, p )} & 0 \\ 0 & e^{-is \Theta (\Lambda , p )} \end{pmatrix}.
\end{equation}

$\Theta$ can be chosen to be a Borel-measurable function on $H \times X_0$ by Eqs.~(\ref{eq:4.1}), (\ref{eq:4.2}), and (\ref{eq:4.4}). If we unwind the definition of the representation space in Table~\ref{tab:1} using the Wigner phase, we obtain
\begin{subequations}\label{eq:4.6}
\begin{equation}\label{eq:4.6a}
\mathcal{H}_{L,s} = L^2 (X_0 ; \mu) \otimes \mathbb{C}^2
\end{equation}
on which we have, for $\psi = \begin{pmatrix} \psi_+ \\ \psi_- \end{pmatrix} = \psi_+ | + \rangle + \psi_- |- \rangle \in \mathcal{H}_{L,s}$ (cf. Eq.~(\ref{eq:1.4})) and $(a, \Lambda) \in G$,
\begin{align}\label{eq:4.6b}
[ U_{L,s} (a, \Lambda) \psi] (p) = e^{-i \langle p, a \rangle} \begin{pmatrix} e^{is \Theta (\Lambda, \Lambda^{-1} p)} \psi_+ ( \Lambda^{-1} p ) \\ e^{-is \Theta (\Lambda, \Lambda^{-1} p)} \psi_- ( \Lambda^{-1} p ) \end{pmatrix} \nonumber \\
= e^{-i \langle p , a \rangle} \sum_{k = \pm} e^{i k s \Theta (\Lambda, \Lambda^{-1} p)} \psi_k (\Lambda^{-1} p) | k \rangle.
\end{align}
\end{subequations}

This representation, which is equivalent to $\Pi_s $ according to Theorem~\ref{theorem:4.1}, is precisely what has been used in the physics literature to describe a massless particle with spin-$s$ (cf. \cite{weinberg}).

\subsection{A brief survey on the RQI of massless particles}\label{sec:4.2}

Shortly after the pioneering papers \cite{peres2002, gingrich2002} that dealt with massive particles with spin-1/2 in the context of RQI appeared, similar treatments of massless particles were taken up in \cite{peres2003, gingrich2003, caban2003}. These three papers considered a communication scenario in which a photon, a massless particle with spin-$1$, is exploited as a qubit carrier just like the way an electron was treated in \cite{peres2002, gingrich2002}. More precisely, they conceived of a protocol in which the sender encodes a qubit of information in the helicity degrees of freedom expressed by the $\mathbb{C}^2$-component of Eq.~(\ref{eq:4.6a}) and examined a situation where the receiver might not be in the same motion state as the sender.

Consequently, just as in the massive particle case (cf. \cite{lee2022b}~Sect.~3.1), the momentum-independent measurement of the helicity of photons was the primary concern of these papers. So, they considered, for each state $\psi \in \mathcal{H}_{L,1}$, the density matrix $\rho$ corresponding to $\psi$
\begin{equation}\label{eq:4.7}
\rho := | \psi \rangle \langle \psi | \in L^2 (X_0 \times X_0 ; \mu \times \mu) \otimes M_2 (\mathbb{C})
\end{equation}
and defined the \textit{helicity reduced density matrix of $\psi$} as
\begin{equation}\label{eq:4.8}
\tau := \textup{Tr}_{L^2} \rho \in M_2 (\mathbb{C})
\end{equation}
where $\textup{Tr}_{L^2}$ is the partial trace with respect to the $L^2(X_0 , \mu)$-component of the Hilbert space.

Although \cite{caban2003} later found a class of single-photon states that have a Lorentz-covariant reduced density matrix $\tau$, the influential paper \cite{peres2003} remarked that, in general, $\tau$ has no transformation law (how it is transformed according to frame changes) because of the momentum dependent Wigner phase in Eq.~(\ref{eq:4.6b}) and suggested an effective helicity measurement scheme which has since been used throughout the RQI literature that dealt with photon (see, for example, the concluding remark of \cite{bradler2011}).\footnote{Despite its importance, we will not explain this scheme since it is beyond the present paper's primary concern.} \cite{lindner2003} went further to claim that $\tau$ is useless even when a single observer is concerned.

Also, the paper \cite{gingrich2003} considered a two-photon state $\psi' \in \mathcal{H}_{L,1} \otimes \mathcal{H}_{L,1}$, formed the density matrix $\rho' := |\psi' \rangle \langle \psi' |$, and defined the \textit{helicity reduced density matrix of the two-particle state $\psi'$} as
\begin{equation}\label{eq:4.9}
\tau' := \textup{Tr}_{L^2} \rho' \in M_2 (\mathbb{C}) \otimes M_2 (\mathbb{C}),
\end{equation}
which was then exploited to study the helicity entanglement of the two photons. The authors used an entanglement measure called \textit{log negativity} on this matrix and showed that the helicity entanglement does indeed vary depending on the motion state of an inertial observer.

\begin{remark}\label{remark:4.2}
The study of massless particles in the context of RQI is still an active area of research (cf. \cite{peres2003, caban2003, gingrich2003, lindner2003, terno2005, caban2007, bradler2008, landulfo2010, bradler2011, bradler2014, rembielinski2018a, rembielinski2018b, many2020, nagele2020, hoffmann2020}). However, to our best knowledge, all the papers that have ever appeared in the literature have used the representation space Eq.~(\ref{eq:4.6}) for the description of massless particles.

Although the survey on the RQI of massless particles given here isn't as detailed as that on the massive particle case given in \cite{lee2022b}~Sect.~3.1, we notice that the same perplexities posed in the massive particle case are still present in the massless particle case. As in \cite{lee2022b}~Sect.~3.2, we claim that these perplexities arise due to an inherent problem of the representation space Eq.~(\ref{eq:4.6}) concerning the concept of \textit{relativistic perception} (see \cite{lee2022b}~Sect.~2.5 for a precise definition of this concept), which can be seen most clearly when we adopt the vector bundle point of view to describe the representation.
\end{remark}

\subsection{The vector bundle point of view for the boosting bundles}\label{sec:4.3}

As in \cite{lee2022b}~Sect.~6.3, we consider two inertial observers Alice and Bob, whose classical observations are related by $(a, \Lambda) \in G$, i.e.,
\begin{equation}\label{eq:4.10}
x_B ^\mu = a ^\mu + \kappa(\Lambda)^\mu _ \nu x_A ^\nu
\end{equation}
where $x_A, x_B \in \mathbb{R}^4$ are Alice's and Bob's coordinate systems, respectively (cf. \cite{lee2022b}~Definition~2.7 and \cite{lee2022b}~Eq.~(2.11)). If Alice has prepared a particle in the state $\psi \in \mathcal{H}_{L,s}$, then Bob would perceive this particle as in the state $U_{L,s} (a,\Lambda) \psi \in \mathcal{H}_{L,s}$ according to the principle of Special Relativity (SR) (cf.\cite{lee2022b}~Sect.~2.3).

They also have the bundle $E_{L,s}$ for the description of the particle at their disposal (cf. \cite{lee2022b}~Sect.~6.3). For the above transformation law for wave functions to be true, Alice's bundle description $E_{L,s} ^A$ and Bob's bundle description $E_{L,s} ^B$ should be related by the $G$-action in Table~\ref{tab:1} as follows
\begin{gather}
\lambda_{L,s} (a, \Lambda ) : E_{L,s} ^A \rightarrow E_{L,s} ^B \nonumber \\
(p,v)^A \mapsto \left(  \Lambda p, e^{- i (\Lambda p)_\mu a^\mu} \tilde{\eta}_{s} \left(W_L ( \Lambda, p) \right) v \right)^B . \label{eq:4.11}
\end{gather}
so that $U_{L,s} (a, \Lambda) \psi = \lambda_{L,s} \circ \psi \circ \Lambda^{-1}$ holds.

By inserting $ (p_0 , v )^A$ and $(a, \Lambda) = (0, L(p))$ into Eq.~(\ref{eq:4.11}), we obtain
\begin{equation}\label{eq:4.12}
\lambda_{L,s} (0, L(p)) (p_0 , v)^A = (p, v)^B.
\end{equation}

As we will see in Sect.~\ref{sec:6.3}, for $s=1$, this equation implies that each fiber $(E_{L,1})_p$ of the boosting bundle does not correctly reflect the perception of a fixed inertial observer who is using this bundle for the description of the particle.\footnote{To see in what sense this statement is true, one is encouraged to read the argument of \cite{lee2022b}~Sect.~3.2.}

More precisely, we will see that $v \in \mathbb{C}^2$ in $(p, v) \in (E_{L,1})_p$ gets meaningful only in an inertial frame $L(p)^{-1}$-transformed with respect to the fixed inertial observer. Hence, if $\psi = f \xi $ with $f:X_0 \rightarrow \mathbb{C}$ a Schwartz class function and $\chi : X_0 \rightarrow \mathbb{C}^2$ a continuous field of unit vectors, then the helicity reduced density matrix (Eq.~(\ref{eq:4.8})) becomes
\begin{equation}\label{eq:4.13}
\tau = \int_{X_0} |f(p)|^2 \chi(p) \chi(p)^\dagger d \mu(p),
\end{equation}
which is an illegitimate operation in that each $\chi(p) \chi(p)^\dagger \in M_2 ( \mathbb{C})$ lives in a different vector space depending on $p \in X_0$ and hence cannot be summable over different values of $p$. (cf. \cite{lee2022b}~Sect.3.2, especially, Eq.~(3.15) there).

Therefore, by using this vector bundle point of view, we have in a sense "proved" that the helicity density matrix Eq.~(\ref{eq:4.8}) indeed has no meaning at all, as claimed earlier in \cite{peres2003}. A similar remark holds for the helicity reduced density matrix for two-particle state Eq.~(\ref{eq:4.9}).

As in \cite{lee2022b}~Sect.~6, we, therefore, conclude that the perplexities posed in the early works of the RQI of massless particles arise because the standard representation space Eq.~(\ref{eq:4.6}) is constructed out of a bundle that does not respect the relativistic perception (in the sense of \cite{lee2022b}~Sect.~2.5), which must be taken into account in dealing with the internal quantum states of a moving particle.

\begin{remark}\label{remark:4.3}
This viewpoint was first suggested in \cite{lee2022a} for massive particles with spin-1/2 and subsequently generalized to massive particles with arbitrary spin in \cite{lee2022b}. There, it was used to show the meaninglessness of the \textit{reduced density matrix for spin} and suggested an alternative bundle description called the \textit{perception bundle description} which is free from the problem that the boosting bundle description has. I.e., each fiber of the perception bundles correctly reflects the relativistic perception of an inertial observer who is using the bundles for the description of massive particles. We will develop the massless version of this bundle description from the next section on.
\end{remark}

\section{The perception bundle construction for massless particles}\label{sec:5}

Fix $0 \neq s \in \frac{1}{2} \in \mathbb{Z}$. By the result of Sect.~\ref{sec:3.1}, we cannot apply the framework of \cite{lee2022b}~Sect.~5 to construct the perception bundles for massless particles with spin-$s$.

It turned out that the perception bundle construction depends on a choice of a (non-unitary) continuous representation $\Phi:H \rightarrow GL(V)$ on a finite-dimensional Hilbert space $(V, \langle \cdot, \cdot \rangle_V )$ and an isometric embedding $\epsilon: \mathbb{C}^2 \rightarrow V$ such that $\Phi|_{K}$ 'almost' becomes $\tilde{\eta}_s$ on the subspace $\epsilon(\mathbb{C}^2)$.

Heuristically, this amounts to embed the helicity-up and down states of a massless particle with spin-$s$ as the extremal helicity states (that is, $\pm s$) among all possible helicity states that a massive particle with spin-$s$ can assume. However, as one might expect from the adverb 'almost' in the preceding paragraph, there is certain complication that needs to be addressed, which leads to the concept of \textit{gauge freedom}. It is an important characteristic exhibited by massless particles such as photon and graviton (cf. \cite{weinberg}) and we are going to see that it is an intrinsic property of all massless particles (cf. Sect.~\ref{sec:8}).

\subsection{Perception bundles}\label{sec:5.1}

First, some notations. For a subset $S$ of a vector space $V$, we denote the (complex) linear span of $S$ as $[S] \leq V$. If $W \leq V$, then, for $ v, w \in V$, $v \equiv w \hspace{0.1cm} (\textup{mod } W)$ means $v-w \in W$.

\begin{lemma}\label{lemma:5.1}
Let $(V, \langle \cdot , \cdot \rangle_V )$ be a finite-dimensional Hilbert space, $\Phi: H \rightarrow GL(V)$ a (non-unitary) Lie group representation, and $\epsilon: \mathbb{C}^2 \rightarrow V $ an isometric embedding (w.r.t. the standard inner product on $\mathbb{C}^2$). Denote $R:= \left[ \Phi(K) \epsilon(\mathbb{C}^2 ) \right] \leq V$. Suppose that $W \leq R$ is a $\Phi(K)$-invariant subspace such that $W \oplus \epsilon (\mathbb{C}^2 ) = R$ (orthogonal direct sum) and
\begin{equation}\label{eq:5.1}
\Phi (B) \epsilon (v) \equiv \epsilon \left( \tilde{\eta}_s (B)  v \right)  \quad (\textup{mod } W)
\end{equation}
for $ v \in \mathbb{C}^2$ and $B \in K$.

Then, for each $p = A p_0 \in X_0$, the subspaces
\begin{equation}\label{eq:5.2}
D_p := \Phi(A)W \leq F_p := \Phi(A) R \leq V
\end{equation}
do not depend on the choice of $A \in H$ such that $Ap_0 = p$, and hence $D:= \amalg_{p \in X_0} D_p $ and $F := \amalg_{p \in X_0} F_p$ form well-defined (smooth) vector bundles over $X_0$.
\end{lemma}
\begin{proof}
Suppose $p= A p_0 = B p_0$. Then, $A^{-1} B \in K$ and hence $\Phi(A^{-1} B )W = W$ by the $\Phi(K)$-invariance of $W$, i.e., $\Phi(A) W = \Phi(B)W$. Analogously, since $R$ is by definition $\Phi(K)$-invariant, we have $\Phi(A)R = \Phi(B)R$. Therefore, we see that $D_p := \Phi(A)W$ and $ F_p :=\Phi(A) R$ do not depend on the choice of $A \in H$ such that $Ap_0 = p$.

For any basis $\{ e_i\}$ of $W$ (resp. $R$) and any smooth local section $\tau :U \rightarrow H$ of the bundle $H \rightarrow H/K \cong X_0$, the set of smooth sections $\{ \Phi (\tau(\cdot)) e_i \}$ of the bundle $X_0 \times V$ is a frame for the bundle $D$ (resp. $F$) on $U \subseteq X_0$. Hence, we see that $D$ and $F$ are smooth vector subbundles of $X_0 \times V$ (cf. Ch.~10 of \cite{lee}).
\end{proof}

\begin{theorem}\label{theorem:5.2}
Consider the setting of Lemma~\ref{lemma:5.1}. Form a quotient vector bundle $E:=F/D$ over $X_0$ and denote the quotient bundle map as $\mathcal{Q} : F \rightarrow E$. Then, endowed with the structures listed in Table~\ref{tab:2}, $F$ becomes a $G$-vector bundle, $E$ an Hermitian $G$-bundle, and $\mathcal{Q}$ a $G$-equivariant map. Also, the induced representation $U$ associated with $E$ is unitarily equivalent to $\Pi_s$.
\end{theorem}

\begin{table}[h]

\caption{The perception bundle construction for massless particles with spin-\texorpdfstring{$s$}{TEXT}}
\label{tab:2}

\centering
\begin{tabular}{|m{1.3cm}|m{4.9cm}|m{6cm}|}
\hline\noalign{\smallskip}
  & $F$  (The potential bundle) & $E$ (The perception bundle) \\
\noalign{\smallskip}\hline\noalign{\smallskip}
Bundle  &  $F $ & $F/D$ \\
\noalign{\smallskip}\hline\noalign{\smallskip}
Metric &  None & $h_{p} (z + D_p, w + D_p) = $ \newline $ \big\langle \Phi(\Lambda)^{-1} z + W , \Phi(\Lambda)^{-1} w + W \big\rangle_{R/W}$ \\
\noalign{\smallskip}\hline\noalign{\smallskip}
Action &   $\vartheta (a,\Lambda) (p, z) $ \newline $ = (\Lambda p, e^{-i \langle \Lambda p, a \rangle} \Phi(\Lambda) z )$ & $\lambda (a, \Lambda) (p , \xi ) $ \newline $ = (\Lambda p,  e^{-i \langle \Lambda p , a \rangle} \overline{\Phi(\Lambda)} \xi)$ \\
\noalign{\smallskip}\hline\noalign{\smallskip}
Space & $\mathcal{B} = \mathcal{Q}^{-1} (\mathcal{H}) $ & $\mathcal{H}= L^2 \left (X_0 , E ; \mu , h \right)$ \\
\noalign{\smallskip}\hline\noalign{\smallskip}
Repn & $ \mathcal{U} (a,\Lambda ) A = \vartheta (a,\Lambda) \circ A \circ \Lambda^{-1}$ & $U(a, \Lambda) \phi = \lambda(a, \Lambda) \circ \phi \circ \Lambda^{-1}$ \\
\noalign{\smallskip}\hline
\end{tabular}

\end{table}

\begin{proof}
The proof is given in Appendix~\ref{sec:C}.
\end{proof}

In the first column of Table~\ref{tab:2}, $\mathcal{B}:= \mathcal{Q}^{-1} ( \mathcal{H})$ is the space of all Borel measurable sections of $F \leq X_0 \times V$ that descend to square integrable sections of $E$ when composed with the quotient map $\mathcal{Q}$ and $\mathcal{U}$ is a (non-unitary, non-continuous) representation of $G$ on this space.

In the second column, $\overline{\Phi(\Lambda)}$ is the map
\begin{gather}
\overline{\Phi(\Lambda)} : E_{p} \rightarrow E_{\Lambda p} \nonumber \\
z+D_p \mapsto \Phi(\Lambda) z + D_{\Lambda p}, \label{eq:5.3}
\end{gather}
which is a well-defined vector bundle automorphism of $E$.

We call the Hermitian $G$-bundles $E$ obtained in this way the \textit{perception bundles for massless particles with spin-$s$}. The subbundles $F$ and $D$ of Lemma~\ref{lemma:5.1} will be respectively called the \textit{potentail bundle} and the \textit{gauge freedom associated with the representation $\Phi$} for reasons that will become clear in Sects.~\ref{sec:6}--\ref{sec:7}.

\subsection{Perception bundles in terms of potential bundles and gauge freedoms}\label{sec:5.2}

The definition of perception bundles as quotient vector bundles causes some difficulty in dealing with wave functions living on the bundles. So, it is expedient to express perception bundles in terms of potential bundles and gauge freedoms that have just been defined.

For $(p, \xi) \in E$, there exists $(p, z) \in F$ such that $\mathcal{Q} (p, z) = (p, \xi)$. Also, we know that when descended to $E$ via $\mathcal{Q}$, two elements $(p, z_1 ), (p, z_2) \in F$ represent the same element if and only if there is $(p, y) \in D$ such that $z_1 - z_2 = y$. We adopt the habit of referring to an element $(p, \xi) \in E$ by only specifying its lift $(p, z) \in F$ with the understanding that $z$ is only defined up to an addition of an element $y \in D_p$. In this convention, the action $\vartheta$ replaces the role of $\lambda$ by the $G$-equivariance of $\mathcal{Q}$.

Now, fix $\phi \in \mathcal{H}$. By using a (smooth) partition of unity, we can always choose a Borel measurable section $\psi \in \mathcal{B}$ such that $\mathcal{Q} \circ \psi = \phi$. Also, two sections $\psi_1, \psi_2 \in \mathcal{B}$ represent the same wave function in $\mathcal{H}$ if and only if there is a (Borel measurable) section $\varphi: X_0 \rightarrow D$ such that $\varphi = \psi_1 - \psi_2$. Analogously as for the bundle elements, we adopt the habit of referring to a wave function $\phi \in \mathcal{H}$ by only specifying its lift $\psi \in \mathcal{B}$ with the understanding that $\psi$ is defined only up to an addition of a (Borel measurable) section $\varphi : X_0 \rightarrow D$. In this convention, the representation $\mathcal{U}$ replaces the role of $U$, i.e.,
\begin{equation}\label{eq:5.4}
\big[ \mathcal{Q} \circ \big(\mathcal{U}(a,\Lambda) \psi \big) \big] (p) = [U(a,\Lambda) \phi] (p)
\end{equation}
by the $G$-equivariance of $\mathcal{Q}$.

\begin{remark}\label{remark:5.3}
The formulae for $\vartheta$ and $\mathcal{U}$ in Table~\ref{tab:2} tell us that even though the bundle $F$ and the space $\mathcal{B}$ are not directly related to the massless single-particle representation $\Pi_s$, they have definite transformation laws and hence one can compare the elements of $F$ and $\mathcal{B}$, respectively, as perceived by different inertial observers as long as one quotients out the effect of the gauge freedom $D$.
\end{remark}

\section{Bundle theoretic descriptions for photon}\label{sec:6}

In this section, we apply the construction of Sect.~\ref{sec:5} to obtain the perception bundle description for photon, a massless particle with spin-1.

\subsection{The perception bundle for photon}\label{sec:6.1}

We take $V_{e} := \mathbb{C}^4$ with the standard inner product structure on which the Minkowski representation $\Phi_e :H \xrightarrow{\kappa} SO^{\uparrow} (1,3) \hookrightarrow GL( V_e)$ is given. Fix an isometric embedding $\epsilon_e : \mathbb{C}^2 \rightarrow V_e$ defined on the basis elements as
\begin{equation}\label{eq:6.1}
\epsilon_e ( | \pm \rangle ) = \frac{1}{\sqrt{2}} \begin{pmatrix} 0 \\ 1 \\ \pm i \\ 0 \end{pmatrix}.
\end{equation}

We write $\epsilon_\pm := \epsilon_e (| \pm \rangle)$. They represent the right-handed and the left-handed circular polarizations of an electromagnetic wave propagating along the $\hat{\mathbf{z}}$-direction in the Minkowski space, respectively.\footnote{Multiply Eq.~(\ref{eq:6.1}) by the plane wave $e^{i (kz - \omega t) }$ propagating along the $\hat{z}$-axis and take the real parts of the expression. How does the resulting vector change as $t$ increases? (cf. Ch.~9 of \cite{griffiths})}

With these choices, we apply Lemma~\ref{lemma:5.1}.
\begin{lemma}\label{lemma:6.1}
Let $W_e := \mathbb{C} p_0$. Then, we have $R_e:= \left[ \Phi_e (K) \epsilon_e (\mathbb{C}^2) \right] = W_e \oplus \epsilon_e (\mathbb{C}^2) \leq V_e$ (orthogonal direct sum), $W_e$ is $\Phi_e(K)$-invariant, and
\begin{equation}\label{eq:6.2}
\Phi_e(B)\epsilon_e(v) \equiv \epsilon_e(\tilde{\eta}_1 (B) v) \quad (\textup{mod } W_e)
\end{equation}
for $v \in \mathbb{C}^2$ and $B \in K$.
\end{lemma}
\begin{proof}
Since $p_0 = (1,0,0,1)$, we see $W_e \perp \epsilon_e (\mathbb{C}^2) $ in $\mathbb{C}^4$. The $\Phi_e (K)$-invariance of $W_e$ is immediate from the definition of $K$ as the little group $H_{p_0}$.

Let $B = \begin{pmatrix} z & b \\ 0 & \overline{z} \end{pmatrix} \in K$. Then, using Eq.~(\ref{eq:1.3}), we find
\begin{align}\label{eq:6.3}
\Phi_e(B) \epsilon_e (| \pm \rangle) &= \Phi_e (B) \begin{pmatrix} 0 \\ 1 \\  \pm i \\ 0 \end{pmatrix} = \kappa(B) \begin{pmatrix} 0 \\ 1 \\ 0 \\ 0 \end{pmatrix} \pm i \kappa(B) \begin{pmatrix} 0 \\ 0 \\ 1 \\ 0 \end{pmatrix}
 \nonumber \\ 
&= \kappa \begin{pmatrix} z & b \\ 0 & \overline{z} \end{pmatrix} \begin{pmatrix} 0 \\ 1 \\ 0 \\ 0 \end{pmatrix} \pm i \kappa \begin{pmatrix} z & b \\ 0 & \overline{z} \end{pmatrix} \begin{pmatrix} 0 \\ 0 \\ 1 \\ 0 \end{pmatrix} 
= \begin{pmatrix} \textup{Re} (\overline{b}z) \\ \textup{Re} (z^2) \\ -\textup{Im} (z^2) \\ \textup{Re} (\overline{b}z) \end{pmatrix} + \pm i \begin{pmatrix} \textup{Im} (\overline{b}z) \\ \textup{Im} (z^2) \\ \textup{Re} (z^2) \\ \textup{Im} (\overline{b} z) \end{pmatrix} \nonumber \\
 &= \begin{cases}
z^2 \epsilon_+ + (\overline{b}z) p_0 &\text{for $+$} \\
\overline{z}^2 \epsilon_- + (b \overline{z}) p_0 &\text{for $-$}
\end{cases}
=\begin{cases}
\epsilon_e \left(\tilde{\eta}_1 (B) | + \rangle \right) + (\overline{b}z) p_0 &\text{for $+$} \\
\epsilon_e \left(\tilde{\eta}_1 (B) | - \rangle \right) + (b \overline{z}) p_0 &\text{for $-$}.
\end{cases}
\end{align}

Therefore, we see that $R_e = W_e \oplus \epsilon_e (\mathbb{C}^2)$ and also Eq.~(\ref{eq:6.2}) holds.
\end{proof}

\begin{proposition}\label{proposition:6.2}
The potential bundle $F_e \leq X_0 \times V_e$ is given by
\begin{equation}\label{eq:6.4}
F_e = \{ (p, z) \in X_0 \times \mathbb{C}^4 : p_\mu z^\mu =0 \}
\end{equation}
and the gauge freedom $D_e \leq F_e$ associated with $\Phi_e$ is given by
\begin{equation}\label{eq:6.5}
D_e = \{ (p, \zeta p) \in X_0 \times \mathbb{C}^4 : \zeta \in \mathbb{C} \}.
\end{equation}
\end{proposition}
\begin{proof}
Let $p = \Lambda p_0 \in X_0$ for some $\Lambda \in H$. Then, for $v \in \mathbb{C}^2$,
\begin{equation*}
p_\mu \big( \Phi_e(\Lambda) \epsilon_e(v) \big)^\mu = \big(\kappa(\Lambda) p_0 \big)_\mu  \big( \Phi_e(\Lambda) \epsilon_e( v) \big)^\mu = (p_0)_\mu \big(\epsilon_e (v) \big)^\mu = 0
\end{equation*}
by the definition of $\Phi_e$. Also,
\begin{equation*}
p_\mu ( \Phi_e(\Lambda) p_0 )^\mu = p_\mu \big(\kappa(\Lambda) p_0 \big)^\mu = p_\mu p^\mu = 0.
\end{equation*}

So, by Lemma~\ref{lemma:6.1},
\begin{equation*}
(F_e)_p :=\Phi(\Lambda) R_e \subseteq \{ z \in \mathbb{C}^4 : p_\mu z^\mu = 0 \}
\end{equation*}
and since the two spaces are both 3-dimesional (the LHS being an isomorphic image of the 3-dimensional space $R_e$ and the RHS being the kernel of a non-zero linear functional $p_\mu ( \hspace{0.1cm} \cdot \hspace{0.1cm} )^\mu : \mathbb{C}^4 \rightarrow \mathbb{C}$), we see that these are in fact equal.

For Eq.~(\ref{eq:6.5}), we just note
\begin{equation*}
(D_e)_p := \Phi_e(\Lambda) \mathbb{C} p_0 = \mathbb{C} \kappa(\Lambda) p_0 = \mathbb{C} p.
\end{equation*}
\end{proof}

Applying Theorem~\ref{theorem:5.2} with the choice $(\epsilon_e, \Phi_e)$, we obtain Table~\ref{tab:3}, the perception bundle description for photon.

\begin{table}[h]

\caption{The perception bundle description for photon}
\label{tab:3}

\centering
\begin{tabular}{|m{1.3cm}|m{5.5cm}|m{7cm}|}
\hline\noalign{\smallskip}
  & $F_e$  (The potential bundle) & $E_e$ (The perception bundle) \\
\noalign{\smallskip}\hline\noalign{\smallskip}
Bundle  &  $F_e$ & $F_e/D_e$ \\
\noalign{\smallskip}\hline\noalign{\smallskip}
Metric &  None & $(h_e)_{p} \big(z + (D_e)_p, w + (D_e)_p \big) = $ \newline $ \big\langle \Phi_e(\Lambda)^{-1} z + W_e , \Phi_e(\Lambda)^{-1} w + W_e \big\rangle_{R_e/W_e}$ \\
\noalign{\smallskip}\hline\noalign{\smallskip}
Action &   $\vartheta_e (a,\Lambda) (p, z) $ \newline $  = (\Lambda p, e^{-i \langle \Lambda p, a \rangle} \Phi_e(\Lambda) z )$ & $\lambda_e (a, \Lambda) (p , \xi ) = (\Lambda p,  e^{-i \langle \Lambda p , a \rangle} \overline{\Phi_e(\Lambda)} \xi)$ \\
\noalign{\smallskip}\hline\noalign{\smallskip}
Space & $\mathcal{B}_e = \mathcal{Q}_e ^{-1} (\mathcal{H}_e) $ & $\mathcal{H}_e = L^2 \left (X_0 , E_e ; \mu , h_e \right)$ \\
\noalign{\smallskip}\hline\noalign{\smallskip}
Repn & $ \mathcal{U}_e (a,\Lambda ) \psi = \vartheta_e (a,\Lambda) \circ \psi \circ \Lambda^{-1}$ & $U_e (a, \Lambda) \phi = \lambda_e (a, \Lambda) \circ \phi \circ \Lambda^{-1}$ \\
\noalign{\smallskip}\hline
\end{tabular}

\end{table}

\subsection{The vector bundle point of view for the perception bundle of photon}\label{sec:6.2}

As in Sect.~\ref{sec:4.3}, if two inertial observers Alice and Bob, who are related by a Lorentz transformation $(a, \Lambda) \in G$ as in Eq.~(\ref{eq:4.10}), are using the perception bundle to describe a photon, then the two observer's descriptions should be related by the action of Table~\ref{tab:3}, i.e.,
\begin{gather}
\lambda_e (a, \Lambda ) : E_e ^{ A} \rightarrow E_e ^{ B} \nonumber \\
(p,\xi)^A \mapsto \left(  \Lambda p, e^{- i (\Lambda p)_\mu a^\mu} \overline{\Phi_e (\Lambda )} \xi \right)^B \label{eq:6.6}
\end{gather}
so that the transformation law for wave functions $U_e (a,\Lambda) \phi = \lambda_e (a,\Lambda) \circ \phi \circ \Lambda^{-1}$ holds.

The convention suggested in Sect.~\ref{sec:5.2} encourages us to rewrite this vector bundle point of view as follows.

\begin{gather}
\vartheta_e (a, \Lambda ) : F_e ^{ A} \rightarrow F_e ^{ B} \nonumber \\
(p,z)^A \mapsto \left(  \Lambda p, e^{- i (\Lambda p)_\mu a^\mu} \Phi_e (\Lambda ) z \right)^B \label{eq:6.6'}\tag{6.6$'$}
\end{gather}
with the understanding that $z$ and $e^{- i (\Lambda p)_\mu a^\mu} \Phi_e (\Lambda ) z$ are defined only up to elements of $(D_e)_p$ and $(D_e)_{\Lambda p}$, respectively.

\subsection{Physical interpretations of the boosting and perception bundle descriptions for photon}\label{sec:6.3}

Observe
\begin{equation}\label{eq:6.7}
\hat{J} ^3 := i ( \tilde{\eta}_1 )_* (J^3) = \begin{pmatrix} 1 & 0 \\ 0 & -1 \end{pmatrix}
\end{equation}
and hence $|\pm \rangle \in \mathbb{C}^2$ represent the helicity-($\pm 1$) states, respectively (cf. Eqs.~(\ref{eq:2.10})--(\ref{eq:2.11}) and Remark~\ref{remark:2.5}). What this means in the case of electromagnetic wave, of which the photon is the quantum (cf. \cite{weinberg}), is that $|\pm \rangle$ represent the right-handed and the left-handed circular polarizations of the wave propagating along the $\hat{\mathbf{z}}$-direction, respectively.

So, the map $\epsilon_e : \mathbb{C}^2 \rightarrow \mathbb{C}^4$ given by Eq.~(\ref{eq:6.1}) is seen to be the Minkowski space realization of the polarizations of a photon propagating along the $\hat{\mathbf{z}}$-direction. So, by definition of $\Phi_e$, the vectors $\Phi_e (\Lambda) \epsilon_\pm$ represent the right-handed and the left-handed circular polarizations of a photon propagating along the $\kappa(\Lambda) \hat{\mathbf{z}}$-direction in $\mathbb{R}^4$, respectively, for each $\Lambda \in H$.

Therefore, we conjecture from Eq.~(\ref{eq:5.2}) that each fiber $(F_e)_p$ consists of all polarization vectors of a photon moving with momentum $p \in X_0$ in a fixed inertial frame. However, since electromagnetic wave does not have a longitudinal polarization (cf. Ch.~2.5 of \cite{weinberg}), the component of the polarization vectors along the direction of the photon's motion state $p \in X_0$ should be quotiented out, which is precisely accomplished by quotienting out the gauge freedom $(D_e)_p$ (cf. Eq.~(\ref{eq:6.5})).

To summarize, each fiber $(E_e)_p$ consists of all \textit{physically realizable} polarization vectors of a photon moving with momentum $p \in X_0$ in a fixed inertial frame whereas $(F_e)_p$ consists of all \textit{mathematically conceivable} polarization vectors. Also, the transformation laws Eqs.~(\ref{eq:6.6})--(\ref{eq:6.6'}) precisely capture the transformation law of (complex) four-vectors, of which the polarization vectors of a photon is an example.

We conclude that each fiber $(E_e)_p$ of the perception bundle $E_e$ correctly reflects the relativistic perception (in the sense of \cite{lee2022b}~Sect.~2.5) of the polarization states of a photon from a fixed inertial observer who is using this bundle for the description of the photon.

In contrast, Eq.~(\ref{eq:4.12}) shows that each fiber $(E_{L,1})_p$ of the boosting bundle for photon does not reflect the relativistic perception of a fixed inertial observer who is using this bundle for the description of a photon. In other words, without recourse to the frame change $L(p)$, the vector $v$ itself does not give Bob's perception of the helicity state since, for example, the right-handed polarization along the $\hat{\mathbf{z}}$-axis of Alice's frame will not appear to be the right-handed polarization along the $\hat{\mathbf{z}}$-axis of the $L(p)$-transformed Bob's frame, whereas Eq.~(\ref{eq:4.12}) forces it. (See \cite{lee2022b}~Sect.~3.2 for a lengthier discussion on this point.)

\subsection{Theoretical implications on the representation}\label{sec:6.4}

\paragraph{The representation space constructed from the perception bundle}

\hfill

From Table~\ref{tab:3}, we see that massless particles with spin-$1$ are described by the following representation, which is equivalent to $\Pi_1$: The space is
\begin{subequations}\label{eq:6.8}
\begin{equation}\label{eq:6.8a}
\mathcal{H}_e:= L^2 ( X_0 , E_e ; \mu , h_e)
\end{equation}
and the representation $U_e : G \rightarrow U (\mathcal{H}_e)$ is given by, for $(a, \Lambda ) \in G$ and $ \phi \in \mathcal{H}_e$,
\begin{equation}\label{eq:6.8b}
[U_e(a,\Lambda)\phi] (p) = e^{-i \langle p , a \rangle} \overline{\Phi_e( \Lambda)} \phi(\Lambda^{-1} p).
\end{equation}
\end{subequations}

Sect.~\ref{sec:5.2} suggests that we consider the representation space
\begin{subequations}\label{eq:6.9}
\begin{equation}\label{eq:6.9a}
\mathcal{B}_e := \mathcal{Q}_e ^{-1} (\mathcal{H}_e)
\end{equation}
and the representation $\mathcal{U}_e : G \rightarrow GL (\mathcal{B}_e)$ defiend as, for $(a, \Lambda ) \in G$ and $ \psi \in \mathcal{B}_e$,
\begin{equation}\label{eq:6.9b}
[\mathcal{U}_e (a,\Lambda)\psi] (p) = e^{-i \langle p , a \rangle} \Phi_e ( \Lambda) \psi(\Lambda^{-1} p)
\end{equation}
\end{subequations}
and express any element $\phi \in \mathcal{H}_e$ by its lift $\psi \in \mathcal{B}_e$ with the understanding that $\psi$ is defined only up to an addition of a Borel-measurable section $\varphi: X_0 \rightarrow D_e$. In this convention, Eq.~(\ref{eq:6.9b}) takes the role of the representation $U_e$.

\paragraph{Maxwell's equations in vacuum, the Lorentz gauge condition, and the gauge freedom of EM as manifestations of relativistic perception}

\hfill

If we define the four-momentum operators $P^\mu$ on $\mathcal{H}_e$ as the infinitesimal generators of the translation operators $U_e (b , I )$, we obtain, for $\phi \in \mathcal{H}_e$ and $p  \in X_0$,
\begin{subequations}\label{eq:6.10}
\begin{align}
\left [P^0 \phi \right] (p) &:= \left[ i \frac{\partial}{\partial b^0} U_e (b, I) \phi \right] (p) =  p_0 \phi(p) = p^0 \phi(p) \label{eq:6.10a} \\
\left [P^j \phi \right] (p) &:= \left[- i \frac{\partial}{\partial b^j} U_e (b, I) \phi \right] (p) =  -p_j \phi(p) = p^j \phi(p) \label{eq:6.10b}
\end{align}
\end{subequations}
by Eq.~(\ref{eq:6.8b}), which are unbounded self-adjoint operators on $\mathcal{H}_e$.

Similarly, for each $\psi \in \mathcal{B}_e$ and $p \in X_0$, we define
\begin{subequations}\label{eq:6.11}
\begin{align}
\left [\mathcal{P}^0 \psi \right] (p) &:= \left[ i \frac{\partial}{\partial b^0} \mathcal{U}_e (b, I) \psi \right] (p) =  p_0 \psi(p) = p^0 \psi(p) \label{eq:6.11a} \\
\left [\mathcal{P}^j \psi \right] (p) &:= \left[- i \frac{\partial}{\partial b^j} \mathcal{U}_e (b, I) \psi \right] (p) =  -p_j \psi(p) = p^j \psi(p), \label{eq:6.11b}
\end{align}
\end{subequations}
which, this time, should be understood as derivatives in the distribution sense.\footnote{See \cite{hormander1} for the definitions of distribution and its derivative. Those who are not familiar with these notions can restrict their attention to smooth sections and their derivatives without losing any of what follows.}

Eq.~(\ref{eq:5.4}) implies that for $\psi \in \mathcal{B}_e$ and $\phi = \mathcal{Q}_e \circ \psi \in \mathcal{H}_e$ such that $P^\mu \phi$ is defined,
\begin{equation}\label{eq:6.12}
\big[ \mathcal{Q}_e \circ \big( \mathcal{P}^\mu \psi \big) \big] (p) = [ P^\mu \phi] (p).
\end{equation}
(it can also be checked directly from Eqs.~(\ref{eq:6.10})--(\ref{eq:6.11}).)

Therefore, in the convention of Sect.~\ref{sec:5.2}, the operators $\mathcal{P}^\mu$ on $\mathcal{B}_e$ are just the momentum operators $P^\mu$ on the quantum Hilbert space $\mathcal{H}_e$ as long as each field $\psi \in \mathcal{B}_e$ is understood as defined only up to a Borel-measurable section of $D_e = \amalg_{p \in X_0} \mathbb{C} p$.

Now observe that by Eq.~(\ref{eq:6.4}) and the relation $p_\mu p^\mu = 0$ that holds for all $p \in X_0$, we have
\begin{equation}\label{eq:6.13}
 \mathcal{P}_\mu \mathcal{P}^\mu \psi  = 0
\end{equation}
and
\begin{equation}\label{eq:6.14}
\mathcal{P}_\mu \psi^\mu  = 0
\end{equation}
for ${ }^\forall \psi \in \mathcal{B}_e$. Also, two fields $\psi, \psi' \in \mathcal{B}_e$ represent the same wave function in $\mathcal{H}_e$ if and only if there exists a (Borel-measurable) scalar function $f : X_0 \rightarrow \mathbb{C}$ such that, for ${ }^\forall p \in X_0$,
\begin{equation}\label{eq:6.15}
\psi_\mu - \psi'_\mu  =  \mathcal{P}_\mu f
\end{equation}
by Eq.~(\ref{eq:6.5}).

To arrange Eqs.~(\ref{eq:6.13})--(\ref{eq:6.15}) into more familiar forms of differential equations, let's assume that the sections $\psi, \psi'$ and the function $f$ are smooth and have compact supports.\footnote{These restrictions can be greatly relaxed if we use the theory of tempered distributions. However, for the sake of brevity, we chose to impose these strong assumptions.}

For each $t \in \mathbb{R}$, we define
\begin{gather}
A(t, \mathbf{x}) := \int_{X_0} \exp (-i p^0 t + i \mathbf{p} \cdot \mathbf{x} ) \sqrt{p^0} \psi(p) \frac{d \mu (p)} {(2 \pi)^{\frac{3}{2}} }  \label{eq:6.16} \\
g(t, \mathbf{x}) := \int_{X_0} \exp (-i p^0 t + i \mathbf{p} \cdot \mathbf{x} ) \sqrt{p^0} f(p) \frac{d \mu (p)} {(2 \pi)^{\frac{3}{2}} } , \label{eq:6.17}
\end{gather}
which are smooth functions on $\mathbb{R}^3$ that contain all the information of the section $\psi$ and the function $f$ being the Fourier transforms of them.

Regarding $A, A'$ and $g$ as defined on $\mathbb{R}^4$, they become smooth functions on $\mathbb{R}^4$ and satisfies the following differential equations by Eqs.~(\ref{eq:6.13})--(\ref{eq:6.15}).

\begin{gather}
\partial_\mu \partial^\mu A =0 \label{eq:6.13'} \tag{6.13$'$} \\
\partial_\mu A^\mu = 0 \label{eq:6.14'} \tag{6.14$'$} \\
A_\mu - A'_\mu = \partial_\mu g \label{eq:6.15'} \tag{6.15$'$},
\end{gather}
from which we notice that Eq.~(\ref{eq:6.13}) is \textit{Maxwell's equations in vacuum}, Eq.~(\ref{eq:6.14}) is the \textit{Lorentz gauge condition}, and Eq.~(\ref{eq:6.15}) expresses the \textit{gauge freedom} of the electromagnetic 4-potential (cf. Ch.~12 of \cite{griffiths}), respectively. We have just found that these are characteristics of all fields in $\mathcal{B}_e$. Given the interpretations of the perception bundles $E_e$ presented in Sect.~\ref{sec:6.3}, we find that \textit{Maxwell's equations in vacuum, the Lorentz gauge condition, and the gauge freedom of EM are nothing but manifestations of a fixed inertial observer's perception of the internal quantum states of a photon}. This fact is even more clear if we look once again at the definition of the bundles $F_e$, $D_e$, and $E_e$ given in Eqs.~(\ref{eq:6.4})--(\ref{eq:6.5}) with the viewpoint of Sect.~\ref{sec:5.2}. These not only are satisfied by the fields in $\mathcal{B}_e$ but also manifest themselves on the level of elements in the fibers of the perception bundle $E_e$ as perceived by a fixed inertial observer (See \cite{lee2022b}~Remark~6.2 for more on this point).

It is remarkable that one discovers all these fundamental equations of EM solely on the basis of the axioms of QM on which the principle of SR is in action.

\begin{remark}\label{remark:6.3}
The fact that the gauge freedom of EM can be derived from the little group of massless particles has been reported earlier in \cite{kupersztych1976, kim1981, mauro2009}. Eq.~(\ref{eq:6.3}) in this paper may be regarded as a reproduction of their results.

The perception bundle description of photon as presented in this section is similar to the Gupta-Bleuler quantization (\cite{gupta1950, bleuler1950}) of electromagnetic fields since both include the process of quotienting out the longitudinal polarization $D_e$.
\end{remark}

\section{Bundle theoretic descriptions for graviton}\label{sec:7}

In this section, we apply the construction of Sect.~\ref{sec:5} to obtain the perception bundle description for graviton, a massless particle with spin-2.

\subsection{The perception bundle for graviton}\label{sec:7.1}
We take $V_g := V_e \otimes V_e = \mathbb{C}^4 \otimes \mathbb{C}^4$ endowed with the product inner product structure induced from the standard inner product of $\mathbb{C}^4$ (i.e., the unique sesquilinear map on $V_g \otimes V_g$ satisfying $( v_1 \otimes w_1 , v_2 \otimes w_2 ) \mapsto (v_1 \cdot v_2 ) (w_1 \cdot w_2 )$) and fix an isometric embedding $\epsilon_g : \mathbb{C}^2 \rightarrow V_g$ defined on the basis elements as
\begin{equation}\label{eq:7.1}
\epsilon_g ( | \pm \rangle) = \epsilon_\pm \otimes \epsilon_\pm.
\end{equation}

They represent the right-handed and the left-handed circular polarizations of a gravitational wave propagating along the $\hat{\mathbf{z}}$-direction in the Minkowski space, respectively (cf. Eq.~(\ref{eq:7.10}) and Ch.~\textrm{IX}.4 of \cite{zee}).

Finally, we take $\Phi_g := \Phi_e \otimes \Phi_e : H \rightarrow GL( V_g)$. With these choices, we apply Lemma~\ref{lemma:5.1}.

\begin{lemma}\label{lemma:7.1}
Let $W_g := \Big[\frac{1}{2} p_0 \otimes p_0 , \frac{1}{2} ( p_0 \otimes \epsilon_+ + \epsilon_+ \otimes p_0), \frac{1}{2} ( p_0 \otimes \epsilon_- + \epsilon_- \otimes p_0) \Big] \leq V_g$, a three-dimensional subspace for which the listed elements form an orthonormal basis. Then, we have $R_g := \left[ \Phi_g (K) \epsilon_g (\mathbb{C}^2) \right] = W_g \oplus \epsilon_g (\mathbb{C}^2) \leq V_g$ (orthogonal direct sum), $W_g$ is $\Phi_g (K)$-invariant, and
\begin{equation}\label{eq:7.2}
\Phi_g (B) \epsilon_g (v) \equiv \epsilon_g ( \tilde{\eta}_2 (B) v ) \quad \textup{mod $W_g$}
\end{equation}
for $v \in \mathbb{C}^2$ and $B \in K$.
\end{lemma}

\begin{proof} Since $\epsilon_+, \epsilon_-$, and $p_0$ are mutually orthogonal vectors in $V_e = \mathbb{C}^4$, we see $W_g \perp \epsilon_g (\mathbb{C}^2)$ and the listed elements in the definition of $W_g$ are orthonormal, and hence $W_g$ is 3-dimensional. The $\Phi_g (K)$-invariance of $W_g$ is immediate from the definition of $K$ as the little group $H_{p_0}$.

Let $B = \begin{pmatrix} z & b \\ 0 & \overline{z} \end{pmatrix} \in K$. Then, by Eq.~(\ref{eq:6.3}), we find
\begin{align}\label{eq:7.3}
\Phi_g(B) & \epsilon_g ( | \pm \rangle )= \Phi_g(B) (\epsilon_\pm \otimes \epsilon_\pm) = \Phi_e(B) \epsilon_\pm \otimes \Phi_e (B) \epsilon_\pm \nonumber \\
&=\begin{cases}
\left(z^2 \epsilon_e ( |+\rangle) + (\overline{b} z) p_0 \right) \otimes \left(z^2 \epsilon_e (|+\rangle) + (\overline{b} z) p_0 \right) & \textup{for $+$} \\
\left(\overline{z}^2 \epsilon_e (|-\rangle) + (b \overline{z}) p_0 \right) \otimes \left(\overline{z}^2 \epsilon_e (|-\rangle) + (b \overline{z} ) p_0 \right) & \textup{for $-$}
\end{cases} \nonumber \\
&=\begin{cases}
\epsilon_g \left( \tilde{\eta}_2 (B) | + \rangle \right) + (\overline{b} z^3 ) ( p_0 \otimes \epsilon_+ + \epsilon_+ \otimes p_0 ) + (\overline{b}z )^2 p_0 \otimes p_0 & \textup{for $+$} \\
\epsilon_g \left( \tilde{\eta}_2 (B) | - \rangle \right) + (b\overline{z}^3) ( p_0 \otimes \epsilon_- + \epsilon_- \otimes p_0 ) + (b\overline{z})^2 p_0 \otimes p_0 & \textup{for $-$}
\end{cases},
\end{align}
from which we conclude that $R_g \subseteq W_g \oplus \epsilon_g (\mathbb{C}^2)$ and Eq.~(\ref{eq:7.2}) holds.

To show the reverse inclusion, we note that in the preceding equation, we can set $z=1$ while $b$ can be any complex number. So, in particular, setting $b=0$ yields $\epsilon_g (\mathbb{C}^2) \subseteq R_g$. Then, varying $b$ appropriately, we deduce that $p_0 \otimes \epsilon_\pm + \epsilon_\pm \otimes p_0 , p_0 \otimes p_0 \in R_g$ also. Hence, we conclude $R_g = W_g \oplus \epsilon_g (\mathbb{C}^2)$.
\end{proof}

Using the identification $V_g = \mathbb{C}^4 \otimes \mathbb{C}^4 \cong M_4 (\mathbb{C})$, we can denote a generic element $M \in V_g$ as $M^{\mu \nu}$, $\mu, \nu = 0 ,1 , 2, 3$. In this notation, we have
\begin{equation}\label{eq:7.3}
(v \otimes w)^{\mu \nu} = v^\mu w^\nu
\end{equation} for $v, w \in \mathbb{C}^4$ and
\begin{equation}\label{eq:7.5}
[ \Phi_g (\Lambda) M] ^{\mu \nu} = \Phi_e (\Lambda)_\alpha ^{\mu} \Phi_e (\Lambda)_\beta ^\nu M^{\alpha \beta}.
\end{equation}

\begin{proposition}\label{proposition:7.2}
The potential bundle $F_g \leq X_0 \times V_g$ is given by
\begin{equation}\label{eq:7.6}
F_g = \{ (p, M) \in X_0 \times V_g : M^{\mu \nu} = M^{\nu \mu} \hspace{0.1cm} \& \hspace{0.1cm}  p_\mu M^{\mu \nu} = 0 \hspace{0.1cm} \& \hspace{0.1cm} M_\mu ^{\mu} = 0 \}
\end{equation}
and the gauge freedom $D_g \leq F_g$ associated with $\Phi_g$ is given by
\begin{equation}\label{eq:7.7}
D_g = \{ (p, p \otimes x + x \otimes p )  : p \in X_0, \hspace{0.1cm} x \in (F_e)_p \}.
\end{equation}
\end{proposition}
\begin{proof}
Let $p= \Lambda p_0 \in X_0$ for some $\Lambda \in H$. By Lemma~\ref{lemma:7.1}, we have
\begin{align*}
(D_g)_p := \Phi_g (\Lambda) W_g = \Big[ p \otimes p , (p \otimes \Phi_e (\Lambda) \epsilon_\pm + \Phi_e (\Lambda) \epsilon_\pm \otimes p ) \Big] \\
= \Big\{ p \otimes x + x \otimes p  : p \in X_0, \hspace{0.1cm} x \in (F_e)_p \Big\}
\end{align*}
since $(F_e)_p := \Phi_e (\Lambda) R_e = \Phi_e (\Lambda) (\mathbb{C} p_0 \oplus \epsilon_e (\mathbb{C}^2) ) = \mathbb{C} p \oplus \Phi_e (\Lambda) \epsilon_e (\mathbb{C}^2)$. So, Eq.~(\ref{eq:7.7}) is proved.

Observe that
\begin{gather*}
(p \otimes x + x \otimes p )^{\mu \nu} = p^\mu x^\nu + x^\mu p^\nu = (x \otimes p + p \otimes x)^{\nu \mu} \\
p_\mu (p \otimes x + x \otimes p)^{\mu \nu} = p_\mu p^\mu x^\nu + p_\mu x^\mu p^\nu = 0 +0 = 0 \\
(p \otimes x + x \otimes p)_{\mu} ^\mu =2 p^\mu x_\mu = 0
\end{gather*}
by the relation $p_\mu p^\mu = 0 , \hspace{0.1cm} \forall p \in X_0$ and Eq.~(\ref{eq:6.4}). These calculations show that $D_g$ is contained in the RHS of Eq.~(\ref{eq:7.6}).

Also, observe that for $M:=\Phi_g (\Lambda) \epsilon_g ( | \pm \rangle) = \Phi_e(\Lambda) |\pm \rangle \otimes \Phi_e (\Lambda) | \pm \rangle \in V_g$,
\begin{gather*}
M^{\mu \nu} = \big(\Phi_e (\Lambda)  | \pm \rangle \big)^\mu \big(\Phi_e (\Lambda)  | \pm \rangle \big)^\nu = M^{\nu \mu} \\
p_\mu M^{\mu \nu} = p_\mu \big( \Phi_e (\Lambda) | \pm \rangle \big)^\mu \big(\Phi_e (\Lambda)  | \pm \rangle \big)^\nu = 0 \\
M_\mu ^\mu = \eta_{\mu \nu} \Phi_e (\Lambda)^\mu _\alpha \Phi_e (\Lambda)^\nu _\beta | \pm \rangle^\alpha | \pm \rangle^\beta = \eta_{\alpha \beta}  | \pm \rangle^\alpha | \pm \rangle^\beta = -1 - ( \pm i )^{2} = 0,
\end{gather*}
which shows that indeed
\begin{equation*}
(F_g)_p := \Phi_g (\Lambda) R_g = \Phi_g (\Lambda) \big( W_g \oplus \epsilon_g (\mathbb{C}^2) \big) = (D_g)_p \oplus \Phi_g (\Lambda) \epsilon_g (\mathbb{C}^2)
\end{equation*}
is contained in the RHS of Eq.~(\ref{eq:7.6}). To see that the equality holds, note that $(F_g)_p$ is of dimension 5 being an isomorphic image of the 5-dimensional vector space $R_g$ (cf. Lemma~\ref{lemma:7.1}) and also the RHS of Eq.~(\ref{eq:7.6}) is of dimension 5 being the kernel of the surjective linear map $V_g \rightarrow \mathbb{C}^{6 + 4 + 1}$ given by
\begin{equation*}
M \mapsto \begin{pmatrix} (M^{\mu \nu} - M^{\nu \mu})_{0 \leq \mu < \nu \leq 3} \\ (p_\mu M^{\mu \nu} )_{0 \leq \nu \leq 3} \\
M_\mu ^\mu \end{pmatrix}.
\end{equation*}
\end{proof}

Applying Theorem~\ref{theorem:5.2} with the choice $(\epsilon_g, \Phi_g)$, we obtain Table~\ref{tab:4}, the perception bundle description for graviton. One should note that we used $\hat{h}_g$ to denote the Hermitian metric of the perception bundle $E_g$ instead of $h_g$. This is because we want to reserve the letter $h$ to denote the linearized gravity. See Sect.~\ref{sec:7.4}.

\begin{table}[h]

\caption{The perception bundle description for graviton}
\label{tab:4}

\centering
\begin{tabular}{|m{1.3cm}|m{5.5cm}|m{7cm}|}
\hline\noalign{\smallskip}
  & $F_g$  (The potential bundle) & $E_g$ (The perception bundle) \\
\noalign{\smallskip}\hline\noalign{\smallskip}
Bundle  &  $F_g$ & $F_g/D_g$ \\
\noalign{\smallskip}\hline\noalign{\smallskip}
Metric &  None & $(\hat{h}_g)_{p} \big(z + (D_g)_p, w + (D_g)_p \big) = $ \newline $ \big\langle \Phi_g(\Lambda)^{-1} z + W_g , \Phi_g(\Lambda)^{-1} w + W_g \big\rangle_{R_g/W_g}$ \\
\noalign{\smallskip}\hline\noalign{\smallskip}
Action &   $\vartheta_g (a,\Lambda) (p, z) $ \newline $  = (\Lambda p, e^{-i \langle \Lambda p, a \rangle} \Phi_g(\Lambda) z )$ & $\lambda_g (a, \Lambda) (p , \xi ) = (\Lambda p,  e^{-i \langle \Lambda p , a \rangle} \overline{\Phi_g(\Lambda)} \xi)$ \\
\noalign{\smallskip}\hline\noalign{\smallskip}
Space & $\mathcal{B}_g = \mathcal{Q}_g^{-1} (\mathcal{H}_g) $ & $\mathcal{H}_g = L^2 \left (X_0 , E_g ; \mu , h_g \right)$ \\
\noalign{\smallskip}\hline\noalign{\smallskip}
Repn & $ \mathcal{U}_g (a,\Lambda ) h = \vartheta_g (a,\Lambda) \circ h \circ \Lambda^{-1}$ & $U_g (a, \Lambda) \phi = \lambda_g (a, \Lambda) \circ \phi \circ \Lambda^{-1}$ \\
\noalign{\smallskip}\hline
\end{tabular}

\end{table}

\subsection{The vector bundle point of view for the perception bundle of graviton}\label{sec:7.2}

As in Sect.~\ref{sec:4.3}, if two inertial observers Alice and Bob, who are related by a Lorentz transformation $(a, \Lambda) \in G$ as in Eq.~(\ref{eq:4.10}), are using the perception bundle to describe a graviton, then the two observer's descriptions should be related by the action of Table~\ref{tab:4}, i.e.,
\begin{gather}
\lambda_g (a, \Lambda ) : E_g ^{ A} \rightarrow E_g ^{ B} \nonumber \\
(p,\xi)^A \mapsto \left(  \Lambda p, e^{- i (\Lambda p)_\mu a^\mu} \overline{\Phi_g (\Lambda )} \xi \right)^B \label{eq:7.8}
\end{gather}
so that the transformation law for wave functions $U_g (a,\Lambda) \phi = \lambda_g (a,\Lambda) \circ \phi \circ \Lambda^{-1}$ holds.

The convention suggested in Sect.~\ref{sec:5.2} encourages us to rewrite this vector bundle point of view as follows.

\begin{gather}
\vartheta_g (a, \Lambda ) : F_g ^{ A} \rightarrow F_g ^{ B} \nonumber \\
(p,z)^A \mapsto \left(  \Lambda p, e^{- i (\Lambda p)_\mu a^\mu} \Phi_g (\Lambda ) z \right)^B \label{eq:7.8'}\tag{7.8$'$}
\end{gather}
with the understanding that $z$ and $e^{- i (\Lambda p)_\mu a^\mu} \Phi_g (\Lambda ) z$ are only defined up to elements of $(D_g)_p$ and $(D_g)_{\Lambda p}$, respectively.

\subsection{Physical interpretations of the boosting and perception bundle descriptions for graviton}\label{sec:7.3}

Observe
\begin{equation}\label{eq:7.9}
\hat{J}^3 := i ( \tilde{\eta}_2 )_* (J^3) = \begin{pmatrix} 2 & 0 \\ 0 & -2 \end{pmatrix}
\end{equation}
and hence $| \pm \rangle \in \mathbb{C}^2$ represent the helicity-($\pm 2$) states, respectively (cf. Eqs.(\ref{eq:2.10})--(\ref{eq:2.11}) and Remark~\ref{remark:2.5}). What this means in the case of gravitational wave is that $|\pm \rangle$ represent the right-handed and the left-handed circular polarizations of the wave propagating along the $\hat{\mathbf{z}}$-direction, respectively.

On the other hand, note that under the identification $V_g \cong M_4 (\mathbb{C})$, we have
\begin{equation}\label{eq:7.10}
| \pm \rangle \otimes | \pm \rangle \cong \begin{pmatrix} 0 & 0 & 0 & 0 \\ 0 & 1 & \pm i & 0 \\ 0 & \pm i & -1 & 0 \\ 0 & 0 & 0 & \end{pmatrix}.
\end{equation}

According to the standard picture of how a gravitational wave, of which the graviton is the quantum, affects a ring of massive particles as the wave passes by (cf. pp.565--567 of \cite{zee}), we notice that $\epsilon_g (|\pm \rangle)$ represent the right-handed and the left-handed circular polarizations of the wave propagating along the $\hat{\mathbf{z}}$-direction, respectively.

Therefore, the map $\epsilon_g : \mathbb{C}^2 \rightarrow \mathbb{C}^4 \otimes \mathbb{C}^4$ given by Eq.~(\ref{eq:7.1}) is seen to be the Minkowski space realization of the polarizations of a graviton propagating along the $\hat{\mathbf{z}}$-direction. So, in view of the definition of $\Phi_g$ and the transformation property of the 2-tensors on the Minkowski space, of which the polarization tensors of a gravitational wave are examples, we see the 2-tensors $\Phi_g (\Lambda) \epsilon_{g} (| \pm \rangle)$ represent the right-handed and the left-handed circular polarizations of a graviton propagating along the $\kappa(\Lambda) \hat{\mathbf{z}}$-direction in $\mathbb{R}^4$, respectively, for each $\Lambda \in H$.

Thus, as in Sect.~\ref{sec:6.3}, we conclude that each fiber $(E_g)_p$ consists of all \textit{physically realizable} polarization 2-tensors of a graviton moving with momentum $p \in X_0$ in a fixed inertial frame whereas $(F_g)_p$ consists of all \textit{mathematically conceivable} polarization 2-tensors. Also, the transformation laws Eqs.~(\ref{eq:7.8})--(\ref{eq:7.8'}) precisely capture the transformation law of (complex) 2-tensors on the Minkowski space.

We conclude that each fiber $(E_g)_p$ of the perception bundle $E_g$ correctly reflects the relativistic perception of polarization states of a graviton from a fixed inertial observer who is using this bundle for the description of a graviton.

In contrast, the same analysis as in Sect.~\ref{sec:6.3} regarding Eq.~(\ref{eq:4.12}) shows that each fiber $(E_{L,2})_p$ of the boosting bundle for graviton does not respect the relativistic perception of inertial observers.

\subsection{Theoretical implications on the representation}\label{sec:7.4}
\paragraph{The representation space constructed from the perception bundle}

\hfill

From Table~\ref{tab:4}, we see massless particles with spin-$2$ are described by the following representation, which is equivalent to $\Pi_2$: The space is
\begin{subequations}\label{eq:7.11}
\begin{equation}\label{eq:7.11a}
\mathcal{H}_g := L^2 (X_0, E_g ; \mu , \hat{h}_g)
\end{equation}
and the representation $U_g : G \rightarrow U(\mathcal{H}_g)$ is given by, for $(a, \Lambda) \in G$ and $\phi \in \mathcal{H}_g$,
\begin{equation}\label{eq:7.11b}
[U_g (a, \Lambda) \phi] (p) = e^{-i \langle p , a \rangle} \overline{ \Phi_g (\Lambda)} \phi (\Lambda^{-1} p ).
\end{equation}
\end{subequations}

Sect.~\ref{sec:5.2} suggests that we consider the representation space
\begin{subequations}\label{eq:7.12}
\begin{equation}\label{eq:7.12a}
\mathcal{B}_g := \mathcal{Q}_g ^{-1} (\mathcal{H}_g)
\end{equation}
and the representation $\mathcal{U}_g :G \rightarrow GL(\mathcal{B}_g)$ defiend as, for $(a, \Lambda) \in G$ and $\psi \in \mathcal{B}_e$,
\begin{equation}\label{eq:7.12b}
[\mathcal{U}_g (a, \Lambda) \psi ] (p) = e^{- i\langle p, a \rangle} \Phi_e (\Lambda) \psi (\Lambda^{-1} p)
\end{equation}
\end{subequations}
and express any element $\phi \in \mathcal{H}_g$ by its lift $\psi \in \mathcal{B}_g$ with the understanding that $\psi$ is defiend only up to an addition of a Borel-measurable section $\varphi : X_0 \rightarrow D_g$. In this convention, Eq.~(\ref{eq:7.12b}) takes the role of the representation $U_g$.

\paragraph{Einstein's field equations in vacuum, the traceless-transverse gauge condition, and the gauge freedom of GR as manifestations of relativistic perception}

\hfill

As in Sect.~\ref{sec:6.4}, we define the four-momentum operators $P^\mu$ and $\mathcal{P}^\mu$ on $\mathcal{H}_g$ and $\mathcal{B}_g$ as the infinitesimal generators of the translation operators $U_g (b,I)$ and $\mathcal{U}_g (b,I)$, respectively. Then, Eqs.~(\ref{eq:7.11b}) and (\ref{eq:7.12b}) yield
\begin{subequations}\label{eq:7.13}
\begin{equation}\label{eq:7.13a}
[P^\mu \phi] (p) = p^\mu \phi (p)
\end{equation}
\begin{equation}\label{eq:7.13b}
[\mathcal{P}^\mu \psi ] (p) = p^\mu \psi(p)
\end{equation}
\end{subequations}
for ${ }^\forall \phi \in \mathcal{H}_g$ and ${ }^\forall \psi \in \mathcal{B}_g$ and hence
\begin{equation}\label{eq:7.14}
\left[ \mathcal{Q}_g \circ (\mathcal{P}^\mu \psi ) \right] (p) = [P^\mu \phi ] (p)
\end{equation}
whenever $\mathcal{Q}_g \circ \psi = \phi$ and $P^\mu \phi$ is defined.

Therefore, in the convention of Sect.~\ref{sec:5.2}, the operators $\mathcal{P}^\mu$ are just the momentum operators $P^\mu$ on the quantum Hilbert space $\mathcal{H}_g$ as long as each field $\psi \in \mathcal{B}_g$ is understood as defined only up to an addition of a Borel-measurable section of $D_g = \amalg_{p \in X_0} (p \otimes (F_e)_p + (F_e)_p \otimes p)$.

Now observe that by Eq.~(\ref{eq:7.6}) and the relation $p_\mu p^\mu =0$ which holds for all $p \in X_0$, we have
\begin{equation}\label{eq:7.15}
\mathcal{P}_\mu \mathcal{P}^\mu \psi = 0,
\end{equation}
and
\begin{equation}\label{eq:7.16}
\mathcal{P}_{\mu} \psi^{\mu \nu} = 0, \quad \psi_\mu ^\mu = 0, \quad \psi_{\mu \nu} = \psi_{\nu \mu}
\end{equation}
for ${ }^\forall \psi \in \mathcal{B}_g$. Also, two fields $\psi, \psi' \in \mathcal{B}_g$ represent the same wave function in $\mathcal{H}_g$ if and only if there exists a (Borel-measurable) section $A : X_0 \rightarrow F_e$ such that
\begin{equation}\label{eq:7.17}
\psi_{\mu \nu} - \psi'_{\mu \nu} = \mathcal{P}_\mu A_\nu + \mathcal{P}_\nu A_\mu
\end{equation}
by Eq.~(\ref{eq:7.7}).

As in Sect.~\ref{sec:6.4}, let's use the following Fourier transforms, with the assumption that the sections $\psi, \psi'$, and $A$ are smooth and have compact supports.

\begin{gather}
h(t, \mathbf{x}) := \int_{X_0} \exp (-i p^0 t + i \mathbf{p} \cdot \mathbf{x} ) \sqrt{p^0} \psi(p) \frac{d \mu (p)} {(2 \pi)^{\frac{3}{2}} }  \label{eq:7.18} \\
B(t, \mathbf{x}) := \int_{X_0} \exp (-i p^0 t + i \mathbf{p} \cdot \mathbf{x} ) \sqrt{p^0} A(p) \frac{d \mu (p)} {(2 \pi)^{\frac{3}{2}} } \label{eq:7.19}
\end{gather}

These are smooth functions on $\mathbb{R}^4$ and become the Fourier transforms of $\psi$ and $A$, respectively, when restricted to each hypersurface $\{t\} \times \mathbb{R}^3$.

By Eqs.~(\ref{eq:7.15})--(\ref{eq:7.17}), the following equations are satisfied by these fields.

\begin{gather}
\partial_\mu \partial^\mu h =0 \label{eq:7.15'} \tag{7.15$'$} \\
\partial_\mu h^{\mu \nu} = 0, \quad h_\mu ^\mu = 0, \quad h_{\mu \nu} = h_{\nu \mu} \label{eq:7.16'} \tag{7.16$'$} \\
h_{\mu \nu} - h'_{\mu \nu} = \partial_\mu B_\nu + \partial_\nu B_\mu \label{eq:7.17'} \tag{7.17$'$}.
\end{gather}

If we interpret the field $h$ as representing the linearized gravity (i.e., the metric perturbation), then Eq.~(\ref{eq:7.15'}) is \textit{Einstein's fields equations in vacuum}, Eq.~(\ref{eq:7.16'}) is the \textit{traceless-transverse gauge condition for symmetric 2-tensors}, and Eq.~(\ref{eq:7.17'}) expresses \textit{the gauge freedom} of the linearized gravity $h$ (cf. Ch.~IX.4 of \cite{zee}). We have just seen that Eqs.~(\ref{eq:7.15})--(\ref{eq:7.17}) are characteristics of all fields in $\mathcal{B}_g$. Given the interpretation of the perception bundle $E_g$ presented in Sect.~\ref{sec:7.3}, we find that \textit{Einstein's field equations in vacuum, the traceless-transverse gauge condition, and the gauge freedom of GR are nothing but manifestations of a fixed inertial observer's perception of internal quantum states of a graviton.} This fact is even more clear if we look once again at the definition of the bundles $F_g$, $D_g$, and $E_g$ given in Eqs.~(\ref{eq:7.6})--(\ref{eq:7.7}) with the viewpoint of Sect.~\ref{sec:5.2}. These not only are satisfied by the fields in $\mathcal{B}_g$ but also manifest themselves on the level of elements in the fibers of the perception bundle $E_g$ as perceived by a fixed inertial observer.

\section{The perception bundle description for general massless particles; gauge freedom}\label{sec:8}

In this section, we apply the construction of Sect.~\ref{sec:5} to massless particles with arbitrary spin. We fix $0 < s \in \frac{1}{2} \mathbb{Z}$ throughout this section.

\subsection{The perception bundle for massless particles with spin-\texorpdfstring{$s$}{TEXT}}\label{sec:8.1}

Following the heuristics laid out at the beginning of Sect.~\ref{sec:5}, we take $V_s := \Sigma^{2s} (\mathbb{C}^2)$, the $2s$-symmetric tensor product space endowed with the innerproduct induced from $\mathbb{C}^2$, and $\Phi_s : H \rightarrow GL( V_s)$ given by $T \mapsto \Sigma^{2s} (T)$, the $2s$-symmetric tensor product of the map $T \in H \leq GL( \mathbb{C}^2 )$. For details about this construction, see \cite{lee2022b}~pp.25--27.

Since
\begin{equation}\label{eq:8.1}
\left\{ \sqrt{\frac{(2s)!}{k! (2s - k )!}} | + \rangle^{k} | - \rangle^{2s -k} : 0 \leq k \leq 2s \right\}
\end{equation}
is an orthonormal basis for the $(2s +1)$-dimensional space $V_s$ consisting of the eigenvectors of the operator $\hat{J}^3:= i (\Phi_s )_* (J^3)$ with eigenvalues $k-s$ for $0 \leq k \leq 2s$ (cf. \cite{lee2022b}~Theorem~4.7), the most natural choice for the embedding $\epsilon_s : \mathbb{C}^2 \rightarrow V_s$ would be
\begin{equation}\label{eq:8.2}
\epsilon_s( | \pm \rangle ) := | \pm \rangle^{2s} \in V_s.
\end{equation}

Note that $\Phi_s$ is a (non-unitary) representation and $\epsilon_s : \mathbb{C}^2 \rightarrow V_s$ is an isometric embedding due to the orthonormality of Eq.~(\ref{eq:8.1}). With these choices, we apply Lemma~\ref{lemma:5.1}.

\begin{lemma}\label{lemma:8.1}
Let $W_s := \left[ | + \rangle^{k} | - \rangle^{2s -k} : 1 \leq k \leq 2s-1 \right] \leq V_s$. Then, we have $R_s := \left[ \Phi_s (K) \epsilon_s (\mathbb{C}^2) \right] = W_s \oplus \epsilon_s (\mathbb{C}^2) = V_s$ (orthogonal direct sum), $W_s$ is $\Phi_s (K)$-invariant, and
\begin{equation}\label{eq:8.3}
\Phi_s (B) \epsilon_s (v) \equiv \epsilon_s ( \tilde{\eta}_s (B) v ) \quad \textup{(mod $W_s$)}
\end{equation}
for $v \in \mathbb{C}^2$ and $B \in K$.
\end{lemma}
\begin{proof}
By the orthogonality of the vectors in Eq.~(\ref{eq:8.1}), we have $W_s \perp \epsilon_s (\mathbb{C}^2)$.

Let $B = \begin{pmatrix} z & b \\ 0 & \overline{z} \end{pmatrix} \in K$. Observe that

\begin{align*}
\Phi_s (B) \epsilon_s (| \pm \rangle)= \Phi_s (B) \left(| \pm \rangle^{2s} \right) = (B | \pm \rangle )^{2s} = \begin{cases}
z^{2s} |+ \rangle^{2s} & \textup{for $+$} \\
( b |+ \rangle + \overline{z} | - \rangle )^{2s} & \textup{for $-$}
\end{cases} \\
=\begin{cases}
z^{2s} | + \rangle^{2s} & \textup{for $+$} \\
\overline{z}^{2s} | - \rangle^{2s} + \sum_{l=1} ^{2s} \binom{2s}{l} b^{l} \overline{z}^{2s - l} | + \rangle^{l} | - \rangle^{2s -l} & \textup{for $-$}
\end{cases} \\
=\begin{cases}
\epsilon_s \big( \tilde{\eta}_s (B) |+ \rangle \big) & \textup{for $+$} \\
\epsilon_s \big( \tilde{\eta}_s (B) |- \rangle \big) + \sum_{k=1} ^{2s} \binom{2s}{l} b^{k} \overline{z}^{2s - k} | + \rangle^{k} | - \rangle^{2s -k} & \textup{for $-$}
\end{cases},
\end{align*}
which shows that $R_s \subseteq W_s \oplus \epsilon_s (\mathbb{C}^2)$ and Eq.~(\ref{eq:8.3}) holds.

Since in the preceding equation $b$ can be any complex number while we can set $z =1$, we see that
\begin{equation*}
|+ \rangle^k | - \rangle^{2s - k} \in R_s \quad \textup{for $0 \leq k \leq 2s $},
\end{equation*}
which implies $W_s \oplus \epsilon_s (\mathbb{C}^2) \subseteq V_s \subseteq R_s$.
\end{proof}

\begin{proposition}\label{proposition:8.2}
The potential bundle $F_s \leq X_0 \times V_s$ is given by
\begin{equation}\label{eq:8.4}
F_s = X_0 \times V_s
\end{equation}
and the gauge freedom $D_s \leq F_s$ associated with $\Phi_s$ is given by
\begin{equation}\label{eq:8.5}
D_s = \left\{ (\Lambda p_0 , z) \in F_s : \Lambda \in H \textup{ and } z \in \Big[ (\Lambda | + \rangle) ^k (\Lambda | - \rangle)^{2s - k} : 1 \leq k \leq 2s -1 \Big] \right\}.
\end{equation}
\end{proposition}
\begin{proof}
The proof is immediate from the definitions Eq.~(\ref{eq:5.2}) and Lemma~\ref{lemma:8.1}.
\end{proof}

Applying Theorem~\ref{theorem:5.2} with the choice $(\epsilon_s, \Phi_s)$, we obtain Table~\ref{tab:5}, the perception bundle description for massless particles with spin-$s$.

\begin{table}[h]

\caption{The perception bundle description for massless particles with spin-\texorpdfstring{$s$}{TEXT}}
\label{tab:5}

\centering
\begin{tabular}{|m{1.3cm}|m{5.5cm}|m{7cm}|}
\hline\noalign{\smallskip}
  & $F_s$  (The potential bundle) & $E_s$ (The perception bundle) \\
\noalign{\smallskip}\hline\noalign{\smallskip}
Bundle  &  $X_0 \times V_s $ & $F_s/D_s$ \\
\noalign{\smallskip}\hline\noalign{\smallskip}
Metric &  None & $(h_s)_{p} \big(z + (D_s)_p, w + (D_s)_p \big) = $ \newline $ \big\langle \Phi_s(\Lambda)^{-1} z + W_s , \Phi_s(\Lambda)^{-1} w + W_s \big\rangle_{R_s/W_s}$ \\
\noalign{\smallskip}\hline\noalign{\smallskip}
Action &   $\vartheta_s (a,\Lambda) (p, z) $ \newline $  = (\Lambda p, e^{-i \langle \Lambda p, a \rangle} \Phi_s(\Lambda) z )$ & $\lambda_s (a, \Lambda) (p , \xi ) = (\Lambda p,  e^{-i \langle \Lambda p , a \rangle} \overline{\Phi_s(\Lambda)} \xi)$ \\
\noalign{\smallskip}\hline\noalign{\smallskip}
Space & $\mathcal{B}_s = \mathcal{Q}_s ^{-1} (\mathcal{H}_s) $ & $\mathcal{H}_s = L^2 \left (X_0 , E_s ; \mu , h_s \right)$ \\
\noalign{\smallskip}\hline\noalign{\smallskip}
Repn & $ \mathcal{U}_s (a,\Lambda ) A = \vartheta_s (a,\Lambda) \circ A \circ \Lambda^{-1}$ & $U_s (a, \Lambda) \phi = \lambda_s (a, \Lambda) \circ \phi \circ \Lambda^{-1}$ \\
\noalign{\smallskip}\hline
\end{tabular}

\end{table}

\subsection{The vector bundle point of view for the perception bundle of massless particles with spin-\texorpdfstring{$s$}{TEXT}}\label{sec:8.2}

As in Sect.~\ref{sec:4.3}, if two inertial observers Alice and Bob, who are related by a Lorentz transformation $(a, \Lambda) \in G$ as in Eq.~(\ref{eq:4.10}), are using the perception bundle to describe a massless particle with spin-$s$, then the two observer's descriptions should be related by the action of Table~\ref{tab:5}, i.e.,
\begin{gather}
\lambda_s (a, \Lambda ) : E_s ^{ A} \rightarrow E_s ^{ B} \nonumber \\
(p,\xi)^A \mapsto \left(  \Lambda p, e^{- i (\Lambda p)_\mu a^\mu} \overline{\Phi_s (\Lambda )} \xi \right)^B \label{eq:8.6}
\end{gather}
so that the transformation law for wave functions $U_s (a,\Lambda) \phi = \lambda_s (a,\Lambda) \circ \phi \circ \Lambda^{-1}$ holds.

The convention suggested in Sect.~\ref{sec:5.2} encourages us to rewrite this vector bundle point of view as follows.

\begin{gather}
\vartheta_s (a, \Lambda ) : F_s ^{ A} \rightarrow F_s ^{ B} \nonumber \\
(p,z)^A \mapsto \left(  \Lambda p, e^{- i (\Lambda p)_\mu a^\mu} \Phi_s (\Lambda ) z \right)^B \label{eq:8.6'}\tag{8.6$'$}
\end{gather}
with the understanding that $z$ and $e^{- i (\Lambda p)_\mu a^\mu} \Phi_s (\Lambda ) z$ are only defined up to elements of $(D_s)_p$ and $(D_s)_{\Lambda p}$, respectively.

\subsection{Physical interpretations of the boosting and perception bundle descriptions for massless particles with spin-\texorpdfstring{$s$}{TEXT}}\label{sec:8.3}

Since by definition $\epsilon_s ( | \pm \rangle ) = | \pm \rangle^{2s}$, which are eigenvectors of the operator $\hat{J}^3 := i (\Phi_s )_* (J^3)$ with eigenvalues $\pm s$, respectively, we see that $\epsilon_s (| \pm \rangle)$ can be interpreted as the internal angular momentum (i.e., spin) realization of the helicity states $| \pm \rangle \in \mathbb{C}^2$ of a massless particle with spin-$s$. So, we may call each vector in $\epsilon ( \mathbb{C}^2) \subseteq V_s$ a \textit{spin state} of the particle.

In \cite{lee2022b}~Sect.~6.4.1, we saw that, as internal quantum states of a massive particle with spin-$s$, the vectors $| \pm \rangle^{2s}$ represent eigenstates of the spin along the $\hat{\mathbf{z}}$-direction with eigenvalues $\pm s$ in a particle rest frame. Also in the same section, we saw that the vectors $\Phi_s (\Lambda) (| \pm \rangle^{2s}) $ represent eigenstates of the spin along the $\kappa(\Lambda) \hat{\mathbf{z}}$-direction with eigenvalues $\pm s$ in an inertial frame where the particle is moving with momentum $p$.

Thus, in analogy with Sect.~\ref{sec:6.3}, we hold the interpretation that Eqs.~(\ref{eq:8.4})--(\ref{eq:8.5}) imply that while each fiber $F_p$ consists of all \textit{mathematically conceivable} spin states of a spin-$s$ massless particle with momentum $p$, each fiber $E_p$ consists of all \textit{physically realizable} spin states of a spin-$s$ massless particle with momentum $p$.\footnote{Of course, this interpretation should be confirmed by experiments. While we are convinced that massless particles can assume only $\pm s$-spin eigenstates along its direction of momentum (the fact that $\Pi_s$ is unitarily equivalent to the induced representation $U_s$ associated with the bundle $E_s$ might be considered as a mathematical proof for it), we could not have found a source that plainly states this fact.}

With this interpretation, we see that the argument of \cite{lee2022b}~Sect.~6.4.1 proves that Eqs.~(\ref{eq:8.6})--(\ref{eq:8.6'}) are precisely the transformation law for the spin states of spin-$s$ massless particles. Therefore, we conclude that each fiber $(E_s)_p$ of the perception bundle $E_s$ correctly reflects the relativistic perception of a fixed inertial observer who is using this bundle for the description of a massless particle with spin-$s$.

In contrast, Eq.~(\ref{eq:4.12}) again shows that each fiber $(E_{L,s})_p$ of the boosting bundle does not respect the relativistic perception of inertial observers.

\subsection{Theoretical implications on the representation; gauge freedom}\label{sec:8.4}

From Table~\ref{tab:5}, we see massless particles with spin-$s$ are described by the following representation: The space is
\begin{subequations}\label{eq:8.8}
\begin{equation}\label{eq:8.8a}
\mathcal{H}_s := L^2 (X_0 , E_s ; \mu , h_s)
\end{equation}
and the representation $U_s:G \rightarrow U(\mathcal{H}_s)$ is given by, for $(a, \Lambda) \in G$ and $\phi \in \mathcal{H}_s$,
\begin{equation}\label{eq:8.8b}
[U_s (a, \Lambda) \phi ] (p) = e^{-i \langle p , a \rangle} \overline{ \Phi_s (\Lambda)} \phi (\Lambda^{-1} p ).
\end{equation}
\end{subequations}

Sect.~\ref{sec:5.2} suggests that we consider the representation space
\begin{subequations}\label{eq:8.9}
\begin{equation}\label{eq:8.9a}
\mathcal{B}_s : = \mathcal{Q}_s ^{-1} (\mathcal{H}_s)
\end{equation}
and the representation $\mathcal{U}_s : G \rightarrow GL (\mathcal{B}_s)$ defined as, for $(a, \Lambda) \in G$ and $A \in \mathcal{B}_s$,
\begin{equation}\label{eq:8.9b}
[\mathcal{U}_s (a, \Lambda) A] (p) = e^{-i \langle p , a \rangle} \Phi_s (\Lambda) A (\Lambda^{-1} p)
\end{equation}
\end{subequations}
and express any element $\phi \in \mathcal{H}_s$ by its lift $\psi \in \mathcal{B}_s$ with the understanding that $\psi$ is defined only up to an addition of a Borel-measurable section $\varphi: X_0 \rightarrow D_s$. In this convention, Eq.~(\ref{eq:8.9b}) takes the role of the representation $U_s$.

\paragraph{Gauge freedom as a manifestation of relativistic perception}

\hfill

As in Sect.~\ref{sec:6.4}, we define the four-momentum operators $P^\mu$ and $\mathcal{P}^\mu$ on $\mathcal{H}_s$ and $\mathcal{B}_s$ as the infinitesimal generators of the translation operators $U_s (b,I)$ and $\mathcal{U}_s (b,I)$, respectively. Then, Eqs.~(\ref{eq:8.8b}) and (\ref{eq:8.9b}) yield
\begin{subequations}\label{eq:8.10}
\begin{equation}\label{eq:8.10a}
[P^\mu \phi] (p) = p^\mu \phi (p)
\end{equation}
\begin{equation}\label{eq:8.10b}
[\mathcal{P}^\mu \psi ] (p) = p^\mu \psi(p)
\end{equation}
\end{subequations}
for ${ }^\forall \phi \in \mathcal{H}_s$ and ${ }^\forall \psi \in \mathcal{B}_s$ and hence
\begin{equation}\label{eq:8.11}
\left[ \mathcal{Q}_s \circ (\mathcal{P}^\mu \psi ) \right] (p) = [P^\mu \phi ] (p)
\end{equation}
whenever $\mathcal{Q}_s \circ \psi = \phi$ and $P^\mu \phi$ is defined.

Therefore, in the convention of Sect.~\ref{sec:5.2}, the operators $\mathcal{P}^\mu$ are just the momentum operators $P^\mu$ on the quantum Hilbert space $\mathcal{H}_s$ as long as each field $\psi \in \mathcal{B}_s$ is understood as defined only up to an addition of a Borel-measurable section of $D_s$.

Now observe that by Eq.~(\ref{eq:7.6}) and the relation $p_\mu p^\mu =0$ which holds for all $p \in X_0$, we have
\begin{equation}\label{eq:8.12}
\mathcal{P}_\mu \mathcal{P}^\mu \psi = 0
\end{equation}
and also two fields $\psi, \psi' \in \mathcal{B}_s$ represent the same wave function in $\mathcal{H}_s$ if and only if there exists a (Borel-measurable) section $\varphi : X_0 \rightarrow D$ such that
\begin{equation}\label{eq:8.13}
\psi - \psi' = \varphi.
\end{equation}

Unlike the case of photon and graviton, however, we don't have a classical analogue with which we can compare Eq.~(\ref{eq:8.12})\footnote{One might argue that it is the Klein-Gordon equation. However, that equation is also satisfied by the sections of the boosting bundle by Eq.~(\ref{eq:4.6b}) and hence cannot be classified as a manifestation of relativistic perception. In the cases of photon and graviton, it was the overall properties of the fields $A$ and $h$ expressed by Eqs.~(\ref{eq:6.13})--(\ref{eq:6.15}) and Eqs.~(\ref{eq:7.15})--(\ref{eq:7.17}), respectively, that enabled us to conclude that Eqs.~(\ref{eq:6.13}) and (\ref{eq:7.15}) are Maxwell's equations and Einstein's field equations in vacuum, respectively.} and an explicit formula that expresses the gauge freedom Eq.~(\ref{eq:8.13}) since each fiber $D_p$ does not have an explicit defining formula depending only on $p \in X_0$.

In spite of these limitations, the analysis of this section shows one important aspect of massless particles. Namely, \textit{massless particles with spin-$s$ exhibit $(2s -1)$-dimensional gauge freedom\footnote{This agrees with the results of Sects.~\ref{sec:6}--\ref{sec:7}.} when one takes into account the perception of the internal quantum states of the particles with respect to a fixed inertial observer}. This gauge freedom manifests itself even in the level of elements in the fibers of the perception bundle $E_s$. Note that this aspect was not visible in the boosting bundle description given in Sect.~\ref{sec:4}.

Also, we remark that massless particles with spin-1/2 don't have any gauge freedom because $2 \cdot \frac{1}{2} - 1 =0$.

\appendix

\def\thesection{\Alph{section}}

\section{Proof of Theorem~\ref{theorem:3.3}}\label{sec:A}

\subsection{Preliminaries}\label{sec:A.1}

First, we need some facts regarding connections and curvatures on principal fiber bundles. The definitions and elementary properties of them (as well as the notations used in the following) can be found in Ch.~6 of \cite{tu}. Throughout this subsection, $G$ is a Lie group, $\mathfrak{g}$ is its Lie algebra, $P \xrightarrow{\pi} M$ is a principal $G$-bundle, $\eta : G \rightarrow GL(V)$ is a Lie group representation, and $E:= P \times_\eta V$ is the vector bundle associated with $\eta$.

We denote by $\Omega^k (M,E)$ the space of $E$-valued smooth $k$-forms on $M$ and
\begin{equation}\label{eq:A.1}
\Omega_\eta ^k (P,V) := \left\{ \phi \in \Omega^k (P,V) : \text{$\phi$ is horizontal and } (r_g)^* \phi = \eta(g^{-1} ) \cdot \phi \right\},
\end{equation}
the set of all $V$-valued smooth tensorial $k$-forms of type $\eta$. (Here, being horizontal means that $\phi (X_1 , \cdots , X_k)$ vanishes if one of its argument $X_i$ is vertical, i.e., $\pi_* (X_i) = 0$.)

There is a canonical isomorphism between the two vector spaces. Given $\phi \in \Omega_\eta ^k (P,V)$, define, for $x \in M$ and $p \in \pi^{-1} (x)$,
\begin{equation}\label{eq:A.2}
\phi_x ^\flat  (v_1 , \cdots, v_k) = \big[p , \phi_p (\tilde{v}_1 , \cdots , \tilde{v}_k ) \big]
\end{equation}
where each $\tilde{v}_i \in T_p P$ is a lift of $v_i \in T_x M$, i.e., $\pi_{*, p} (\tilde{v}_i) = v_i$ for $1 \leq i \leq k$. Hence, in particular, when we are given a local section $e : U \rightarrow P$, we have
\begin{equation}\label{eq:A.3}
\phi_x ^\flat (v_1 , \cdots , v_k) = \big[e(x) , (e^{*} \phi)_x ( v_1 , \cdots , v_k) \big]
\end{equation}
for $ x \in U$ and $ v_i \in T_x U$.

Also, given $\psi \in \Omega^k (M,E)$, define, for $ p \in P$,
\begin{equation}\label{eq:A.4}
\psi_p ^\sharp (u_1 , \cdots, u_k) = [ p, \cdot \hspace{0.1cm}]^{-1} \Big( ( \pi^* \psi)_p (u_1 , \cdots , u_k) \Big).
\end{equation}
(See \cite{lee2022b}~Eq.~(A.1) for a discussion of the isomorphism $[p, \cdot \hspace{0.1cm}]:V \rightarrow E_{\pi(p)}$.)

\begin{theorem}[\cite{tu}, Theorem~31.9]\label{theorem:A.1}
The map $\flat : \Omega_\eta ^k (P,V) \rightarrow \Omega^k (M,E)$ is a well-defined linear isomorphism with inverse $\sharp$. 
\end{theorem}

Note that \cite{lee2022b}~Proposition~A.4 can be viewed as a special case of this theorem for 0-forms.

\begin{proposition}\label{proposition:A.2}
Suppose we are given a connection $\omega$ on $P$. We define an affine connection $\nabla$ on $E$ by
\begin{equation}\label{eq:A.5}
\begin{tikzcd}[baseline=(current  bounding  box.center), column sep=1.5em]
    \Omega_\eta ^0 (P, V ) \arrow[r, "D"] \arrow{d}{\cong}[swap]{\flat}
    &\Omega_\eta ^1 (P,V) \arrow{d}{\flat}[swap]{\cong}
\\
    \Omega^0 (M, E) \arrow[r, dashed, "\nabla"]  & \Omega^1 (M,E)
    \end{tikzcd}
\end{equation}
where $D \phi = d \phi + (\eta_* \omega) \cdot \phi$ is the covariant derivative (here, $\cdot$ denotes the canonical product between the $\mathfrak{gl}(V)$-valued form $\eta_* \omega$ and the $V$-valued form $\phi$. Cf.\cite{tu}).

Then, each local section $e : U \rightarrow P$ induces a local frame $x \mapsto [e(x), \cdot \hspace{0.1cm}]$ for $E$ (cf. \cite{lee2022b}~Lemma~A.7) and the $\mathfrak{gl}(V)$-valued $1$-form $\theta_e : = \eta_* e^*  \omega : TU \rightarrow \mathfrak{gl}(V)$ is the connection matrix of $\nabla$ relative to this frame.
\end{proposition}
\begin{proof}
Fix a local section $e:U \rightarrow P$. By \cite{lee2022b}~Lemma~A.7, $x \mapsto [e(x), \cdot \hspace{0.1cm}]$ is a local frame for $E$. Let $s \in \Omega^0 (M,E)$ be a smooth section. Then, we have, for $x \in U$,
\begin{equation*}
s(x) = (s^\sharp)^\flat (x) = \big[ e(x), (e^* s^\sharp)_x \big]
\end{equation*}
by Theorem~\ref{theorem:A.1} and Eq.~(\ref{eq:A.3}). Hence, $(e^* s^\sharp) : U \rightarrow V$ is the component function of $s$ with respect to the local frame $x \mapsto [e(x), \cdot \hspace{0.1cm}]$ (readers are invited to expand it in components by choosing a basis for $V$).

Let $X \in \frak{X}(M)$ be a smooth vector field. Then, for $x \in U$,
\begin{align*}
( \nabla_X s )_x = \big( d s^\sharp + (\eta_* \omega ) \cdot s^\sharp \big)_x ^\flat (X) = \Big[ e(x) , d (e^* s^\sharp ) (X)_x + (\eta_* e^* \omega) (X)_x \cdot (e^* s^\sharp)_x \Big] \\
=\Big[e(x), X_x (e^* s^\sharp) + (\eta_* e^* \omega)_x (X_x) \cdot (e^* s^\sharp)_x\Big],
\end{align*}
which shows that $\theta_e = \eta_* e^* \omega$ is indeed the connection matrix of $\nabla$ relative to the local frame $x \mapsto [e(x), \cdot \hspace{0.1cm}]$ by the preceding paragraph. 
\end{proof}

If we endow the frame bundle $Fr(E)$ with a right $G$-action via the representation $\eta:G \rightarrow GL(V)$, the following becomes a $G$-equivariant bundle homomorphism
\begin{equation}\label{eq:A.6}
\begin{tikzcd}[baseline=(current  bounding  box.center), column sep=1.5em]
    \overline{\eta}: P \arrow ["{p \mapsto [p, \cdot \hspace{0.1cm}] }"]{rr} \arrow{dr}[swap]{\pi} 
    & & Fr(E) \arrow[dl, "{\pi_{Fr}}"]
\\
    &M &
    \end{tikzcd}
\end{equation}
which is given by $(x, g) \mapsto (x, \eta (g))$ in local coordinate representations when the local trivializations of $P$ and $E$ are given by $(x,g) \mapsto e(x) g$ and $(x,v) \mapsto [e(x), v ]$, respectively, where $e : U \rightarrow P$ is a local section (cf. \cite{lee2022b}~Eqs.~(A.13)--(A.14)). So, in particular, we have $[p,v] = \overline{\eta} (p) v $ for ${}^\forall [p,v] \in E$. 

Observe that, given a local section $e : U \rightarrow P$, the map $\overline{\eta} \circ e :U \rightarrow Fr(E)$ is a local section of $Fr(E)$ such that
\begin{equation}\label{eq:A.7}
[e(x), v] = (\overline{\eta} e (x)) v,
\end{equation}
i.e., the frame $x \mapsto [e(x), \cdot \hspace{0.1cm}]$ that appeared in Proposition~\ref{proposition:A.2} is equal to $\overline{\eta} \circ e$.

\begin{lemma}\label{lemma:A.3}
The affine connection $\nabla$ on $E$ of Propostion \ref{proposition:A.2} induces a connection $\theta$ on the frame bundle $Fr(E)$ by the procedure of Ch.~29 of \cite{tu}. Then, for each local section $e: U \rightarrow  P$, we have
\begin{equation}\label{eq:A.8}
( \overline{\eta} e )^* \theta = \eta_* e^* \omega.
\end{equation}
\end{lemma}
\begin{proof}
This holds essentially by the definition of $\theta$, Eq.~(\ref{eq:A.7}), and Proposition~\ref{proposition:A.2}. 
\end{proof}

\begin{proposition}\label{proposition:A.4}
The following diagram commutes
\begin{equation}\label{eq:A.9}
\begin{tikzcd}[baseline=(current  bounding  box.center), column sep=1.5em]
    &T Fr (E) \arrow["\theta"]{dr} & \\ 
    TP \arrow["\overline{\eta}_* "]{ur} \arrow["\omega"]{dr} & &\mathfrak{gl}(V) \\
    &\mathfrak{g} \arrow["\eta_*"]{ur}  &
    \end{tikzcd}
\end{equation}
by virtue of which we have the following equality between the curvatures
\begin{equation}\label{eq:A.10}
\overline{\eta}^* \Omega_{\theta} = \eta_* \Omega_{\omega}.
\end{equation}
\end{proposition}
\begin{proof}
 Let $A \in \mathfrak{g}$ and denote the fundamental vector field on $P$ as $\underline{A}$ (i.e., $\underline{A}_p = \left. \frac{d}{dt} \right|_{t=0} (p \cdot \exp (tA))$ for $p \in P$). Then, by the $G$-equivariance of $\overline{\eta}$, we have 
\begin{equation*}
\overline{\eta}_{*, p} (\underline{A}_p) = \underline{ \eta_* (A) }_{\overline{\eta} (p)}.
\end{equation*}

So, we see
\begin{equation*}
\theta \overline{\eta}_* (\underline{A}) = \theta ( \underline{\eta_* A} ) = \eta_* (A) = \eta_* \omega (\underline{A})
\end{equation*}
since $\theta$ and $\omega$ are connection $1$-forms on $Fr(E)$ and $P$, respectively. So, $\theta \overline{\eta}_* = \eta_* \omega $ on the vertical bundle $\mathcal{V}P \leq TP$.

Now, choose a local section $e : U \rightarrow P$ which induces a local section $\overline{\eta} e : U \rightarrow Fr(E)$. By Eq.~(\ref{eq:A.8}),
\begin{equation*}
 e^*  (\theta \overline{\eta}_* ) =  (\overline{\eta} \circ e )^* \theta = \eta_* e^* \omega = e^* \eta_* \omega
\end{equation*} 
which shows that $\theta \overline{\eta_*} = \eta_* \omega$ on the horizontal bundle over $P$ induced by $\omega$ as well.

For Eq.~(\ref{eq:A.10}), observe that
\begin{align*}
\overline{\eta}^* \Omega_\theta = \overline{\eta}^* \big( d \theta + \frac{1}{2} [ \theta , \theta] \big) =  d (\theta  \overline{\eta}_*) + \frac{1}{2} [ (\theta \overline{\eta}_*) , (\theta \overline{\eta}_*)]  \\
=d (\eta_* \omega) + \frac{1}{2} [ (\eta_* \omega) , (\eta_* \omega)]  = \eta_* \Omega_\omega.
\end{align*}  
\end{proof}

\begin{proposition}\label{proposition:A.5}
Let $f$ be an invariant polynomial on $\mathfrak{gl}(V)$. Then, $f (\Omega_{\theta})$ and $f ( \eta_* \Omega_{\omega} ) $, which are forms on $Fr(E)$ and $P$ respectively, induce the same cohomology class on $M$.
\end{proposition}
\begin{proof}
Since $f(\Omega_{\theta})$ is a basic form, there exists a closed form $[\Lambda] \in H^* (M)$ such that $f(\Omega_{\theta} ) = \pi_{Fr} ^* \Lambda$ (cf. \cite{tu}). By Eqs.~(\ref{eq:A.6}) and (\ref{eq:A.10}),
\begin{equation*}
f( \eta_* \Omega_{\omega} ) = f( \overline{\eta}^* \Omega_\theta)= \overline{\eta}^* f( \Omega_{\theta})= \overline{\eta}^* \pi_{Fr} ^* \Lambda = \pi^* \Lambda
\end{equation*}
which proves the claim. 
\end{proof}

This has the following implication.

\begin{theorem}\label{theorem:A.6}
The characteristic class of $E = P\times_\eta V$ associated with an invariant polynomial $f$ on $\mathfrak{gl}(V)$ is equal to the characteristic class of $P$ associated with the $\textup{Ad}(G)$-invariant polynomial $f \circ \eta_* $ on $\mathfrak{g}$.
\end{theorem}

\subsection{The computation}\label{sec:A.2}

Now, let's embark on the computation of the first Chern classes of the line bundles $\mathcal{E}_s := \mathcal{E}_{\eta_s} = H \times_{\eta_s} \mathbb{C}$ for $s \in \frac{1}{2} \mathbb{Z}$ where $H = SL(2, \mathbb{C})$ and $\eta_s : K \rightarrow U(1)$ is given by Eq.~(\ref{eq:2.7}).

\paragraph{The strategy}

\hfill

Since $K$ is quite complicated, it is not so easy to construct a connection on the principal $K$-bundle $H \rightarrow H/K$. Instead, we consider a simpler subgroup
\begin{equation}\label{eq:A.11}
L := \left\{ \begin{pmatrix} z & 0 \\ 0 & \overline{z} \end{pmatrix} : z \in \mathbb{T} \right\} \leq K \leq H.
\end{equation}

We have a fiber bundle $ \rho : H/L \rightarrow H/ K \cong X_0 ^\pm$ given by $\rho (hL) = hK$ for $h \in H$, whose fiber is $K/L$ (cf. \cite{switzer}, Ch.~4).

\begin{proposition}\label{proposition:A.7}
$\rho$ is a homotopy equivalence.
\end{proposition}

\begin{proof}
First, let's prove that $K/L$ is contractible. Since the map $K \rightarrow E(2)$ given by $\begin{pmatrix} z & b \\ 0 & \overline{z} \end{pmatrix} \mapsto (zb , z^2)$ is a double covering homomorphism with kernel $\{ \pm 1 \}$ (cf. Proposition~\ref{proposition:2.1}), $\mathbb{C}$ becomes a transitive $K$-space with the action
\begin{equation*}
\begin{pmatrix} z & b \\ 0 & \overline{z} \end{pmatrix} \cdot \lambda = zb + z^2 \lambda
\end{equation*}
which is induced from the natural transitive action of $E(2)$ on $\mathbb{C}$. Thus, fixing $\lambda = 0 \in \mathbb{C}$, the map $K \rightarrow \mathbb{C}$ given by $\begin{pmatrix} z & b \\ 0 & \overline{z} \end{pmatrix} \mapsto zb$ is a $K$-equivariant submersion with stabilizer group $L$. So, we see that $K/L \cong \mathbb{C}$, which is contractible.

Therefore, the map $\rho : H/ L \rightarrow H/K$ is a weak homotopy equivalence as can be seen by looking at the long exact sequence of homotopy groups for this fiber bundle. Being a weak homotopy equivalence between two smooth manifolds, it is a homotopy equivalence by Whitehead's theorem (cf. \cite{switzer}). 
\end{proof}

Note that there is a vector bundle morphism
\begin{equation}\label{eq:A.12}
\begin{tikzcd}[baseline=(current  bounding  box.center), column sep=1.5em]
    H\times_{\eta_s |_{L}} \mathbb{C} \arrow{r} \arrow{d} & H \times_{\eta_s} \mathbb{C} \arrow{d} \\ 
    H/L \arrow[r, "\rho"] & H/K 
    \end{tikzcd} 
\end{equation}
given by $[A, z] \mapsto [A, z]$. So, we see that, to obtain the Chern class of $H\times_{\eta_s} \mathbb{C} \rightarrow H/K$, it is sufficient to calculate the Chern class of the bundle $H\times_{\eta_s |_{L}} \mathbb{C} \rightarrow H/L$ and then send it back to $H/K$ via the homotopy equivalence $\rho$.

\paragraph{A connection and a curvature on the principal $L$-bundle \texorpdfstring{$H \rightarrow H/L$}{TEXT}}

\hfill

Let $\theta : TH \rightarrow \mathfrak{h}$ be the Maurer-Cartan form. I.e., $\theta_A (X) := {l_A ^{-1} }_* X = A^{-1} X $ for $X \in T_A H $. Recall from \cite{lee2022b}~Eq.~(4.20) that $\mathfrak{h} := \mathfrak{sl}(2, \mathbb{C}) = \text{span}_{\mathbb{R}} (J^1, J^2 , J^3, K^1 , K^2 , K^3 )$. Note that the Lie algebra $\mathfrak{l}$ of the Lie group $L$ is given by $\mathfrak{l} = \mathbb{R} J^3 \leq \mathfrak{h}$.

Denote the projection map $\beta : \mathfrak{h} \rightarrow \mathfrak{l}$ which is given by $\beta (K^j) = 0 \hspace{0.1cm} {}^\forall j$, $\beta (J^1 ) = \beta (J^2) = 0 $, and $\beta (J^3) = J^3 $. Checking separately the cases $X= K^j$ and $X= J^j$, one can easily see $\beta ( \text{Ad} (B) X ) = \beta (X) = \text{Ad} (B) \beta (X) $ for all $X \in \mathfrak{sl}(2, \mathbb{C})$ and $B \in L$.

Thus, the map $\omega : = \beta \circ \theta : TH \rightarrow \mathfrak{l}$ is a connection on the principal $L$-bundle $H \rightarrow H/L$ by the definition of $\theta$ and the preceding observation. Form a curvature form $\Omega := d \omega + \frac{1}{2} [\omega,\omega]$. Since $[\cdot, \cdot]$ vanishes on the abelian Lie algebra $\mathfrak{l}$ and $d \theta + \frac{1}{2} [ \theta , \theta] = 0$, we see, for $X, Y \in TH$,
\begin{equation}\label{eq:A.13}
\Omega (X,Y) = d \omega (X,Y) = \beta d \theta (X,Y) = - \beta \big( [ \theta (X) , \theta (Y)] \big).
\end{equation}

Therefore, for $X, Y \in \mathfrak{h}$ so that $AX, AY \in T_A H$,
\begin{equation}\label{eq:A.14}
\Omega_A (AX,AY) = -\beta \big([ X , Y ] \big) = 2 \text{Im} (X_{12} Y_{21} - X_{21} Y_{12} ) J^3
\end{equation}
where $X_{ij}, Y_{ij}$ denote their matrix entries as elements of $\mathfrak{h} \leq \mathfrak{gl} (2, \mathbb{C})$.

Note that the representation $\eta_{-1/2} |_L : L \rightarrow U(1)$ is a Lie group isomorphism given by $\begin{pmatrix} z & 0 \\ 0 & \overline{z} \end{pmatrix} \mapsto \overline{z}$ (cf. Eq.~(\ref{eq:2.7})) and hence the map $(\eta_{-1/2} |_L )_* : \mathfrak{l} \rightarrow \mathfrak{u}(1)$ is a Lie algebra isomorphism sending $J^3$ to $\frac{i}{2}$. Therefore, by Theorem~\ref{theorem:A.6}, the first Chern class of the bundle $H \times_{\eta_{-1/2} |_L} \mathbb{C} \rightarrow H/L$ is represented by the following basic form on the principal $L$-bundle $H \xrightarrow{\pi_H} H/L$.
\begin{equation}\label{eq:A.15}
\phi_A (AX, AY) := \frac{i}{2\pi} (\eta_{-1/2})_* \Omega _A (AX, AY) = - \frac{1}{2\pi} \text{Im} ( X_{12} Y_{21} - X_{21} Y_{12} ),
\end{equation}
i.e.,
\begin{equation}\label{eq:A.16}
c_1 (H \times_{\eta_{-1/2} |_L} \mathbb{C}) = [\phi^\flat] \in H^2 (H/L; \mathbb{Z}) \lesssim H_{\text{dR}} ^2 (H/L; \mathbb{R})
\end{equation}
with the $\flat$ understood as indicating the fact that $\phi = \pi_H ^* \phi^\flat$. (This is indeed a special case of Eq.~(\ref{eq:A.2}). See \cite{tu}.)

\paragraph{A cohomology generator of the space \texorpdfstring{$H/K \cong X_0 ^\pm$}{TEXT}}

\hfill

Recall that we are identifying $X_0 ^\pm \cong \mathbb{R}^3 \setminus \{0 \}$ via the map Eq.~(\ref{eq:3.2}). Thus, the $2$-form $\zeta$ (Eq.~(\ref{eq:3.6})) is a generator the cyclic group $H^2 (X_0 ^\pm ; \mathbb{Z}) \cong \textup{Hom} \big( H_2 (X_0 ^\pm ; \mathbb{Z}) , \mathbb{Z} \big) \leq \textup{Hom} \big( H_2 (X_0 ^\pm ; \mathbb{Z}) , \mathbb{R} \big) \cong H_{\text{dR}}^2 (X_0 ^\pm ; \mathbb{R})$ since it gives $1$ when paired with the generating homology class $[\mathbb{S}^2] \in H_2 (X_0 ^\pm ; \mathbb{Z}) \cong \mathbb{Z}$ (that is, when integrated over $\mathbb{S}^2$). The following observation is essential.

\begin{proposition}\label{proposition:A.8}
The $2$-form $\zeta$ is Lorentz invariant. I.e., for ${}^\forall A \in SL(2, \mathbb{C})$,
\begin{equation}\label{eq:A.17}
A^* \zeta = \zeta
\end{equation}
\end{proposition}
\begin{proof}
Since $|p^0 | = |\mathbf{p}|$, we notice that $4 \pi \zeta$ is the standard volume form when restricted to a sphere of any radius. Consequently, we see that $R^* \zeta = \zeta$ for any rotation $R \in SU(2)$. Therefore, since any element $A \in SL( 2, \mathbb{C})$ can be written as a product of three rotations and one boosting along the $z$-axis (cf. \cite{lee2022b}~Eq.~(2.19)), it suffices to check $B^* \zeta = \zeta$ for the boostings along the $z$-axis.

Let $B := e^{u K^3}, \hspace{0.1cm} u \in \mathbb{R}$ be a boosting along the $z$-axis. Denoting $C:= \cosh u$ and $S:= \sinh u$, we have $\kappa \left(B \right) = \begin{pmatrix} C & 0 & 0 & S \\ 0 & 1 & 0 & 0 \\ 0 & 0 & 1 & 0 \\ S & 0 & 0 & C \end{pmatrix} \in SO^\uparrow (1,3)$.

Since $d p^0 = \frac{1}{p^0} \left( p^1 dp^1 + p^2 dp^2 + p^3 dp^3 \right)$, we have
\begin{align*}
B^* \zeta := \kappa(B) ^* \zeta &= \frac{1}{|C p^0 + S p^3|} \Big(  p^1 d p^2 \wedge d(S p^0 + C p^3 ) + p^2 d(S p^0 + C p^3) \wedge d p^1  \\
&\hspace{7cm} + (S p^0 + C p^3) d p^1 \wedge d p^2 \Big) \\
&= \frac{1}{| C p^0 + S p^3 |} \left( C |p^0| \zeta + \frac{S}{p^0}  \big( p_\mu p^\mu \big) dp^1 \wedge dp^2 + \frac{ |p^0|}{p^0} S p^3 \zeta \right) \\
&=\frac{ |p^0|}{p^0} \frac{Cp^0 + S p^3} { | C p^0 + S p^3 |} \zeta = \zeta .
\end{align*} 
\end{proof}

\paragraph{A formula for \texorpdfstring{$\zeta$ on $H$}{TEXT}}

\hfill

To compare the cohomology generator $\zeta$ with the Chern class-representing form Eq.~(\ref{eq:A.15}), we lift $\zeta$ to a form on $H$. Denote the projection map $H \rightarrow H/K \cong X_0 ^\pm$ as $\varphi$, which is thus given by $\varphi (A) = A p_0 ^\pm$.

\begin{lemma}\label{lemma:A.9}
The differential of $ \varphi$ is given by, for $X \in \mathfrak{h}$ (and hence $AX \in T_A H$, $\forall A \in H$),
\begin{equation}\label{eq:A.18}
\varphi_{*,A} (AX) = A x
\end{equation}
where $x \in T_{p_0 ^\pm } X_0 ^\pm \subseteq \mathbb{R}^4$ is given by, in the notation of Eq.~(\ref{eq:1.2}),
\begin{equation}\label{eq:A.19}
\tilde{x} = X (p_0 ^\pm)^\sim + (p_0 ^\pm)^\sim X^\dagger
\end{equation}
and hence $x^0 = x^3 = \pm 2 \textup{Re} X_{11} $, $ x^1 = \pm 2 \textup{Re} X_{21} 
$ and $ x^2 = \pm 2 \textup{Im} X_{21}$.
\end{lemma}

\begin{proof}
By the definition of $\varphi$,
\begin{align*}
\big(\varphi_{*, A} (AX) \big)^\sim = \big( \kappa_* (AX) p_0 ^\pm \big)^\sim = \big( \left. \frac{d}{dt} \right|_{t = 0} \kappa (A (I + tX ) ) p_0 ^\pm \big)^\sim  \\ = A \left(X (p_0 ^\pm)^\sim + (p_0 ^\pm)^\sim X^\dagger \right) A^\dagger = ( \kappa(A) x )^\sim = (A x )^\sim.
\end{align*} 
\end{proof}

\begin{proposition}\label{proposition:A.10}
For $A \in H$ and $X, Y \in \mathfrak{h}$, we have
\begin{equation}\label{eq:A.20}
\left( \varphi^* \zeta \right)_A (AX, AY) =  \frac{1}{2\pi} \textup{Im} ( \overline{X_{21}} Y_{21} - \overline{Y_{21}} X_{21}).
\end{equation}
\end{proposition}
\begin{proof}
Using the notation of the preceding lemma, observe
\begin{align*}
\left( \varphi^* \zeta \right)_A (AX, AY) = \zeta_{A p_0 ^\pm} ( Ax, Ay ) = \left(A^* \zeta \right)_{p_0 ^\pm} (x, y) = \zeta_{p_0 ^\pm} (x,y) \\
 = \frac{1}{4 \pi} (x^1 y^2 - x^2 y^1)
\end{align*}
where we have used Proposition~\ref{proposition:A.8} and Lemma~\ref{lemma:A.9}. By the definitions of $x, y$ in terms of $X,Y$ given in Lemma~\ref{lemma:A.9}, we see this implies
\begin{align*}
\left( \varphi^* \zeta \right)_A (AX, AY) = \frac{1}{\pi} \left( \text{Re} X_{21} \text{Im} Y_{21} - \text{Re} Y_{21} \text{Im} X_{21} \right) = \frac{1}{\pi} \text{Im} (\overline{X_{21}} Y_{21} ) \nonumber \\ =  \frac{1}{2\pi} \text{Im} ( \overline{X_{21}} Y_{21} - \overline{Y_{21}} X_{21}).
\end{align*} 
\end{proof}

\paragraph{Comparison between \texorpdfstring{$\phi$ and $\varphi^* \zeta$}{TEXT}}

\hfill

Let's compare the two $2$-forms $\phi$ (Eq.~(\ref{eq:A.15})) and $\varphi^* \zeta$ (Eq.~(\ref{eq:A.20})) on $H$. The following difference is also a $2$-form on $H$.

\begin{align}\label{eq:A.21}
\psi_A (AX, AY) &:= \big( \varphi^* \zeta - \phi \big)_A (AX, AY) \nonumber \\
 &=  \frac{1}{2 \pi} \text{Im} \left[ (X+ X^\dagger)_{12} Y_{21} - (Y+ Y^\dagger)_{12} X_{21} \right]
\end{align}

We claim that $\psi$ is a closed basic form on the principal $L$-bundle $H \rightarrow H/L$ and $[\psi^\flat] = 0 \in H^2 (H/L ; \mathbb{Z} )$. To show this, we need a lemma.

\begin{lemma}\label{lemma:A.11}
Let $J:= SU(2) \xhookrightarrow{\iota} H$. The map $r: H \rightarrow J$ defined by $ r(A) = (A^{\dagger -1} A^{-1} )^{\frac{1}{2}} A$ induces a deformation retract $\overline{r} : H/L \rightarrow J/L$.
\end{lemma}

\begin{proof}
Consider the map $F : H \times I \rightarrow H$ defined as $F_t (A) = (A^{\dagger -1} A^{-1} )^{t/2} A$ where $I := [0,1]$ and $(A^{\dagger -1} A^{-1})^{t/2}$ is defined by the functional calculus. This map is continuous since, if $(t_i) \subseteq I$ and $(A_i) \subseteq H$ are two sequences converging to $t \in I$ and $A \in H$ respectively, then

\begin{align*}
\|F_{t_i} (A_i) - F_t (A)\| &\leq \|F_{t_i} (A_i) - F_t (A_i) \| + \| F_t (A_i) - F_t (A) \| \\ &\leq \|(A_i ^{\dagger -1} A_i ^{-1} )^{t_i /2} - (A_i ^{\dagger -1} A_i ^{-1} )^{t/2} \| \|A_i \| \\ & \hspace{0.5cm} +
\|(A_i ^{\dagger -1} A_i ^{-1} )^{t /2} A_i - (A ^{\dagger -1} A ^{-1} )^{t/2} A \|
\end{align*}
where $\| \cdot \|$ denotes the operator norm.

The second term tends to zero by the continuity of the functional calculus and the first term tends to zero because $A_i$ is bounded and the spectrums of the operators $(A_i ^{\dagger -1} A_i ^{-1})$ are eventually contained in a fixed compact set $ K \subseteq ( 0, \infty]$ on which the sequence of functions $f_i (x) = x^{t_i / 2}$ converges uniformly to the function $f(x) = x^{t/2}$. Hence, we conclude that $F$ is continuous on $H \times I$.

Note that $F_1 = \iota \circ r $ and $F_0 = \text{id}_H $. Also, $F_t |_J = \text{id}_J$ since $B^{\dagger -1} = B $ for $B \in J$. Finally, $F_t (AB) = (A^{\dagger -1} B^{\dagger -1} B^{-1} A^{-1 })^{t/2} AB = F_t (A)B$ for $A \in H$, $B \in L \leq J$ and hence the map $\overline{F} : H/L \times I \rightarrow  H/L $ defined by the commutative diagram
\begin{equation}\label{eq:A.22}
\begin{tikzcd}[baseline=(current  bounding  box.center), column sep=1.5em]
    H \times I \arrow["F"]{r} \arrow{d} &  H \arrow{d} \\
H/L \times I \arrow["\overline{F}"]{r} & H/L 
    \end{tikzcd} 
\end{equation}
is a well defined map such that $\overline{F}_t |_{J/L} = \text{id}_{J/L}$, $\overline{F}_0 = \text{id}_{H/L}$ and $\overline{F}_1 = \overline{\iota} \circ \overline{r}$, proving the lemma. 
\end{proof}

Denote the natural projections $H \rightarrow H/L $ and $ J \rightarrow J/L$ as $\pi_H$ and $\pi_J$, respectively. Note that $ \rho \circ \pi_H : H \rightarrow H/K$ is equal to the natural projection $\varphi$ considered in Lemma~\ref{lemma:A.9}.

\begin{proposition}\label{proposition:A.12}
$\psi$ is a closed basic form on the principal bundle $H \rightarrow H/L$.
\end{proposition}
\begin{proof}
We know that $\phi$ (Eq.~\ref{eq:A.15}) is a closed basic form, being a characteristic-class-representing form (cf. \cite{tu}). Also, $\varphi^* \zeta = \pi_H ^* ( \rho^* \zeta)$ is a closed basic form being a pull-back of a closed form $\rho^* \zeta$ on $H/L$. Thus, we see $\psi = \varphi^* \zeta - \phi$ is also a closed basic form. 
\end{proof}

So, $\psi$ is pulled down to a closed form $\psi^\flat \in \Omega^2 (H/L)$ such that
\begin{equation}\label{eq:A.23}
\psi = \pi_H ^* (\psi^\flat).
\end{equation}
(cf. Ch.~6 of \cite{tu} for a discussion on basic forms and characteristic classes)

The situation is depicted in the following diagram.
\begin{equation}\label{eq:A.24}
\begin{tikzcd}[baseline=(current  bounding  box.center), column sep=1.5em]
    J  \arrow[r, hookrightarrow, "\iota"] \arrow["\pi_J"]{dd} &  H \arrow["\pi_H"]{dd} \arrow["\psi"]{dr} & \\
 & & \mathbb{R} \\
J/L  \arrow["\overline{\iota}"]{r} & H/L \arrow{ur}[swap]{\psi^\flat} & 
    \end{tikzcd} 
\end{equation}

\begin{theorem}\label{theorem:A.13}
We have $[\psi^\flat] = 0 \in H^2 (H/L ; \mathbb{Z})$. Hence we conclude
\begin{equation}\label{eq:A.25}
[\rho^* \zeta] = c_1 (H \times_{\eta_{-1/2} |_L} \mathbb{C} ) \in H^2 (H/L;\mathbb{Z}).
\end{equation}
\end{theorem}
\begin{proof}
From Eq.~(\ref{eq:A.21}), we see that $\psi_A (AX , AY) = 0 $ for $X , Y \in \mathfrak{j} := \mathfrak{su}(2)$ since all such $X,Y$ are skew-Hermitian. Hence, $\iota^* \psi = 0 $.

Also, since $0 = \iota^* \psi = \iota^* \pi_H ^* \psi^\flat = \pi_J ^* \overline{\iota}^* \psi^\flat$ and $(\pi_J)_*$ is a surjection at each fiber, we have
\begin{equation*}
\overline{\iota}^* \psi^\flat = 0.
\end{equation*}

By virtue of Lemma~\ref{lemma:A.11}, we see $[ \psi^\flat ] = \overline{r}^* \overline{\iota}^* [\psi^\flat] = 0 \in H^2 (H/L ; \mathbb{Z}) $. Hence, by the definition $\psi = \varphi^* \zeta - \phi$, we conclude
\begin{equation*}
[(\varphi^* \zeta)^\flat ] = [\phi^\flat] \in H^2 (H/L).
\end{equation*}

Finally, since $\varphi^* \zeta = \pi_H ^* (\rho^* \zeta)$ and hence $(\varphi^* \zeta)^\flat = \rho^* \zeta$, Eq.~(\ref{eq:A.16}) implies Eq.~(\ref{eq:A.25}).
\end{proof}

\paragraph{The Chern classes}

\hfill

\begin{proof}[Proof of Theorem~\ref{theorem:3.3}]
Since Eq.~(\ref{eq:A.12}) is a vector bundle morphism, we have $\rho^* c_1 (\mathcal{E}_{-1/2})=c_1 (H \times_{\eta_{-1/2} |_L} \mathbb{C} )$. Hence, by Proposition~\ref{proposition:A.7} and Theorem~\ref{theorem:A.13}, we conclude
\begin{equation}\label{eq:A.26}
[\zeta] = c_1 (\mathcal{E}_{-1/2}) \in H^2 (H/K ; \mathbb{Z}).
\end{equation}

Since $(\eta_s)_* (J^3) = -si = -2s (\eta_{-\frac{1}{2}})_* (J^3)$ (cf. Eq.~(\ref{eq:2.10})), Theorem \ref{theorem:A.6} and Eq.~(\ref{eq:A.26}) assert that, letting $\Omega'$ be an arbitrarily chosen connection 2-form on the principal bundle $H \rightarrow H/K$,

\begin{align*}
c_1 (\mathcal{E}_{s} ) = \Big[ \big(\frac{i}{2 \pi} (\eta_s)_* \Omega' \big)^\flat \Big] = -2s \Big[ \big( \frac{i}{2 \pi} (\eta_{- \frac{1}{2}} )_* \Omega' \big)^\flat \Big] \\
 = -2s \hspace{0.1cm} c_1 ( \mathcal{E}_{-1/2} ) = -2s [\zeta] \in H^2 (H/K ; \mathbb{Z}).
\end{align*} 
\end{proof}

With a little more Algebraic Topology, we obtain a classification of the line bundles over $H/K \cong X_0 ^\pm$.

\begin{theorem}\label{theorem:A.14}
The line bundles $\mathcal{E}_s = H \times_{\eta_s} \mathbb{C}$ are pairwise nonisomorphic and exhaust all isomorphism classes of the line bundles over $H/K \cong X_0 ^\pm$. In fact, the Chern class map
\begin{equation*}
c_1 : \textup{Line}(X_0 ^\pm) \rightarrow H^2 (X_0 ^\pm ,\mathbb{Z})
\end{equation*}
is a group isomorphism if the addition in $\textup{Line} (X_0 ^\pm)$ is given by the tensor product operation.
\end{theorem}
\begin{proof}
This follows from the splitting formula for Chern classes (cf. \cite{tu}) and the fact that the classifying space $\mathbb{CP}^\infty$ for complex line bundles is simultaneously the Eilenberg-Maclane space $K(\mathbb{Z}, 2)$. For more details on these concepts, see \cite{switzer}. 
\end{proof}

\section{A rigorous discussion on the parity inversion symmetry}\label{sec:B}

Theorem~\ref{theorem:3.5} is a well-known fact. However, we couldn't have found an easily accessible source that contains a rigorous proof of this fact. So, we record one here. To do that, we first need to define the enlarged symmetry group $\tilde{G}$ (cf. Sect.~\ref{sec:3.2}).

Embed $SL(2, \mathbb{C})$ into $GL(4, \mathbb{C})$ via a Lie group homomorphism $\Phi :SL(2, \mathbb{C}) \hookrightarrow GL(4, \mathbb{C})$ given by
\begin{equation}\label{eq:B.1}
\Phi(\Lambda) = \begin{pmatrix} \Lambda & 0 \\ 0 & \Lambda^{\dagger -1} \end{pmatrix}.
\end{equation}

Define
\begin{equation}\label{eq:B.2}
\tilde{H}:= SL(2, \mathbb{C}) \hspace{0.05cm} \amalg \hspace{0.05cm} \tilde{\mathcal{P}} \hspace{0.1cm} SL(2,\mathbb{C}) \leq GL(4, \mathbb{C})
\end{equation}
where
\begin{equation}\label{eq:B.3}
\tilde{\mathcal{P}} := \begin{pmatrix} 0 & I_2 \\ I_2 & 0 \end{pmatrix} \in GL (4, \mathbb{C}).
\end{equation}

$\tilde{H}$ is a closed subgroup of $GL(4, \mathbb{C})$ consisting of two simply connected components diffeomorphic to $SL(2, \mathbb{C})$ and doubly covers the two-component group $O^\uparrow (1,3) := SO^{\uparrow} (1,3) \hspace{0.05cm} \amalg \hspace{0.05cm} \mathcal{P} \hspace{0.1cm} SO^{\uparrow} (1,3) \leq O(1,3)$ (cf. \cite{lee2022b}~Theorem~2.8). To write out the covering map explicitly, we need to introduce some notations.

For $ x \in \mathbb{R}^4$, we consider a linear embedding $\gamma : \mathbb{R}^4 \rightarrow \mathfrak{gl}(4, \mathbb{C})$ defined by
\begin{equation}\label{eq:B.4}
\gamma(x) := \begin{pmatrix} 0 & \tilde{x} \\ \utilde{x} & 0 \end{pmatrix} \in \mathfrak{gl} (4, \mathbb{C}).
\end{equation}

Note that $\gamma^\mu := \gamma( \eta^{\mu \nu} e_\nu) = \begin{pmatrix} 0 & \tau_\mu \\ \tau^\mu & 0\end{pmatrix}$ is the Weyl representation of the Dirac matrices (cf. \cite{bleecker, folland2008}). (Here $( e_0 , e_1, e_2, e_3) = \big( (1,0,0,0), (0,1,0,0) , (0, 0, 1, 0 ) , (0, 0, 0, 1) \big)$ is the standard basis of $\mathbb{R}^4$.)

Notice that the covering map $\kappa : SL(2,\mathbb{C}) \rightarrow SO^\uparrow (1,3)$ defined in Eq.~(\ref{eq:1.3}) satisfies
\begin{equation*}
\gamma (\kappa (\Lambda) x ) = \begin{pmatrix} 0 & \Lambda \tilde{x} \Lambda^\dagger \\ \Lambda^{\dagger -1} \utilde{x} \Lambda^{-1} & 0 \end{pmatrix} = \Phi(\Lambda) \gamma (x) \Phi(\Lambda)^{-1}
\end{equation*}
for $\Lambda \in SL(2, \mathbb{C})$ and $ x \in \mathbb{R}^4$.

Therefore, the map $\tilde{\kappa} : \tilde{H} \rightarrow O^\uparrow (1,3)$ defined by
\begin{equation}\label{eq:B.5}
\gamma \big( \tilde{\kappa} (A) x \big) := A \gamma (x)  A^{-1}  \quad A \in \tilde{H} , \hspace{0.1cm}  x \in \mathbb{R}^4
\end{equation}
extends the map $\kappa$,\footnote{Here, we are identifying $SL(2, \mathbb{C})$ as a subgroup of $\tilde{H}$ via the embedding $\Phi$. So, more precisely, this means $\tilde{\kappa}\circ \Phi = \kappa$.} is a Lie group homomorphism, and $\tilde{\kappa} (\tilde{\mathcal{P}}) = \mathcal{P}$ since
\begin{equation*}
\gamma( \tilde{\kappa} ( \tilde{\mathcal{P}})  x ) = \tilde{\mathcal{P}} \begin{pmatrix} 0 & \tilde{x} \\ \utilde{x} & 0 \end{pmatrix} \tilde{\mathcal{P}}^{-1} = \begin{pmatrix} 0 & \utilde{x} \\ \tilde{x} & 0 \end{pmatrix} = \begin{pmatrix} 0 & (\mathcal{P} x )^\sim \\ (\mathcal{P} x )_{\sim} & 0 \end{pmatrix} = \gamma ( \mathcal{P} x ).
\end{equation*}
(Recall that $\mathcal{P} := \text{diag} ( 1, -1, -1, -1) \in O(1,3)$.)

From the fact that $\kappa$ is a double covering homomorphism, it is easy to check that $\tilde{\kappa}$ is indeed a double covering homomorphism. This map gives an action of $\tilde{H}$ on $\mathbb{R}^4$ and we define a semidirect product $\tilde{G}:= \mathbb{R}^4 \ltimes \tilde{H}$ via this action. (Hence, for example, $(0, \tilde{\mathcal{P}}) (a, A)  = (\mathcal{P}a , \tilde{\mathcal{P}} A)$ in $\tilde{G}$ for $a \in \mathbb{R}^4$, $A \in \tilde{H}$.)

Obviously, we have a natural Lie group embedding $1 \times \Phi : G \hookrightarrow \tilde{G}$ given by $(a, \Lambda) \mapsto (a, \Phi(\Lambda))$. We have a decomposition into connected components
\begin{equation}\label{eq:B.6}
\tilde{G} = G \hspace{0.05cm} \amalg \hspace{0.05cm}  (0, \tilde{\mathcal{P}}) \cdot G.
\end{equation}

To prove Theorem~\ref{theorem:3.5}, we need some preparation.

\begin{lemma}\label{lemma:B.1}
Recall that we write $\Lambda x := \kappa(\Lambda)x $ for $\Lambda \in SL(2,\mathbb{C})$, $ x \in \mathbb{R}^4$. Then, we have
\begin{equation}\label{eq:B.7}
\mathcal{P} \Lambda x = \Lambda^{\dagger -1} \mathcal{P} x.
\end{equation}
\end{lemma}
\begin{proof}
\begin{align*}
\mathcal{P} \Lambda x = \tilde{\kappa} (\tilde{\mathcal{P}}) \kappa (\Lambda) x = \tilde{\kappa} \left( \tilde{\mathcal{P}} \begin{pmatrix} \Lambda & 0 \\ 0 & \Lambda^{\dagger -1} \end{pmatrix} \right)x
= \tilde{\kappa} \left( \begin{pmatrix} \Lambda^{\dagger -1} & 0 \\ 0 & \Lambda \end{pmatrix} \tilde{\mathcal{P}} \right) x  \\
= \kappa (\Lambda^{\dagger -1} ) \tilde{\kappa} (\tilde{\mathcal{P}}) x = \Lambda^{\dagger -1} \mathcal{P} x.
\end{align*} 
\end{proof}

\begin{lemma}\label{lemma:B.2}
The map $\delta: G \rightarrow G$ given by
\begin{equation}\label{eq:B.8}
\delta( a, \Lambda) = (\mathcal{P} a , \Lambda^{\dagger -1})
\end{equation}
is a Lie group automorphism such that $\delta \circ \delta = \textup{Id}_G$. Thus, if $\pi$ is an irreducible unitary representation of $G$, then so is $\pi \circ \delta$.
\end{lemma}
\begin{proof}
The proof follows from Lemma~\ref{lemma:B.1} and the definition of semidirect product. 
\end{proof}

Note that the parity inversion $\mathcal{P}$ preserves all the $SO^{\uparrow} (1,3)$-orbits in $\mathbb{R}^4$ listed in \cite{lee2022b}~Proposition~4.3. So, we can apply the same technique as in \cite{lee2022b}~Sect.~4 to obtain all irreducible unitary representations of the group $\tilde{G}$. However, this would take us too far afield and therefore we content ourselves with the following theorem which is sufficient to produce all the single-particle state spaces with parity inversion symmetry that will be discussed in the rest of the paper.

\begin{theorem}\label{theorem:B.3}
Let $\pi : G \rightarrow U(\mathcal{H})$ be an irreducible unitary representation of $G$ such that $\pi \circ \delta$ is not unitarily equivalent to $\pi$. Then, the representation $\tilde{\pi} : \tilde{G} \rightarrow U(\mathcal{H} \oplus \mathcal{H})$ defined by
\begin{equation}\label{eq:B.9}
\tilde{\pi} (a, \Lambda) := [\pi \oplus (\pi \circ \delta) ] (a, \Lambda) = \pi(a, \Lambda) \oplus \pi( \mathcal{P} a, \Lambda^{\dagger -1}) \quad \forall (a, \Lambda) \in G
\end{equation}
and
\begin{equation}\label{eq:B.10}
\tilde{\pi} ( (0, \tilde{\mathcal{P}}) (a, \Lambda) ) := \begin{pmatrix} 0 & I_{\mathcal{H}} \\ I_{\mathcal{H}} & 0 \end{pmatrix} \tilde{\pi} (a, \Lambda) \quad \forall (a, \Lambda) \in G
\end{equation}
is an irreducible representation of $\tilde{G}$.
\end{theorem}
\begin{proof}
For each $x \in \tilde{G}$, $\tilde{\pi}(x)$ is a unitary operator on $\mathcal{H} \oplus \mathcal{H}$ since $\pi$ is a unitary representation. $\tilde{\pi}$ is strongly continuous due to the strong continuity of $\pi$. It is a group homomorphism on the subgroup $G$ by Lemma~\ref{lemma:B.2}. To see that $\tilde{\pi}$ is a group homomorphism on the entire group $\tilde{G}$, it is enough to observe the followings: For $ (a,A) , (b,B) \in G$,
\begin{align*}
\tilde{\pi} \left( (a,A) (0,\tilde{\mathcal{P}}) (b,B)  \right) &=\tilde{\pi} \left( a + \mathcal{P} A^{\dagger -1} b , \Phi (A) \tilde{\mathcal{P}} \Phi(B) \right) \\
&= \tilde{\pi} \left( \mathcal{P} (\mathcal{P} a + A^{\dagger -1} b) , \tilde{\mathcal{P}} \Phi(A^{\dagger -1}) \Phi(B)  \right)  \\
&= \tilde{\pi} \left( (0, \tilde{\mathcal{P}}) (\mathcal{P} a + A^{\dagger -1}b ,  A^{\dagger -1} B  \right) \\
&= \begin{pmatrix} 0 & I_{\mathcal{H}} \\ I_{\mathcal{H}} & 0 \end{pmatrix} \tilde{\pi}  ( \mathcal{P}a , A^{\dagger -1} ) \tilde{\pi} (b,B) \\
&= \begin{pmatrix} 0 & I_{\mathcal{H}} \\ I_{\mathcal{H}} & 0 \end{pmatrix} \left[ \pi ( \mathcal{P}a , A^{\dagger -1} ) \oplus \pi ( a, A) \right] \tilde{\pi} (b,B)  \\
&= \left[ \pi(a,A) \oplus \pi ( \mathcal{P}a , A^{\dagger -1} )  \right] \begin{pmatrix} 0 & I_{\mathcal{H}} \\ I_{\mathcal{H}} & 0 \end{pmatrix} \tilde{\pi} (b,B) \\
&= \tilde{\pi} (a,A) \tilde{\pi} \left( ( 0, \tilde{\mathcal{P}}) (b,B) \right)
\end{align*}

and
\begin{align*}
\tilde{\pi} \left( (0, \tilde{\mathcal{P}}) (a,A) (0,\tilde{\mathcal{P}}) (b,B)  \right) &=\tilde{\pi} \left( \mathcal{P} a +  A^{\dagger -1} b , A^{\dagger -1}B \right) \\
&= \tilde{\pi} (\mathcal{P}a , A^{\dagger -1} ) \tilde{\pi} (b, B)  \\
&= \left[\pi(\mathcal{P}a , A^{\dagger -1} ) \oplus \pi(a, A) \right] \tilde{\pi} (b,B) \\
&= \begin{pmatrix} 0 & I_{\mathcal{H}} \\ I_{\mathcal{H}} & 0 \end{pmatrix} \tilde{\pi}  ( a , A ) \begin{pmatrix} 0 & I_{\mathcal{H}} \\ I_{\mathcal{H}} & 0 \end{pmatrix} \tilde{\pi} (b,B) \\
&= \tilde{\pi} \left( (0, \tilde{\mathcal{P}}) (a,A) \right) \tilde{\pi} \left( (0, \tilde{\mathcal{P}}) (b,B) \right).
\end{align*}

Let $T = \begin{pmatrix} T_{11} & T_{12} \\ T_{21} & T_{22} \end{pmatrix} $ be a bounded operator on $\mathcal{H} \oplus \mathcal{H}$ which commutes with $\tilde{\pi} (x)$ for all $x \in \tilde{G}$. So, a posteriori, $T$ commutes with $\tilde{\pi} (0, \tilde{\mathcal{P}}) = \begin{pmatrix} 0 & I_{\mathcal{H}} \\ I_{\mathcal{H}} & 0 \end{pmatrix}$. Hence, we have $T_{11} = T_{22}$ and $T_{12} = T_{21}$, i.e., $T = \begin{pmatrix} T_{11} & T_{12} \\ T_{12} & T_{11} \end{pmatrix}$.

From the fact that $T  \tilde{\pi} (a,A)  = \tilde{\pi} (a,A) T$ for all $(a, A) \in G$, we get
\begin{equation*}
T_{11} \pi(a,A) = \pi (a,A) T_{11} \text{ and } T_{12} [\pi \circ \delta (a,A)] = [\pi ( a, A)] T_{12}, \quad \forall (a,A) \in G.
\end{equation*}

Since $\pi$ and $\pi \circ \delta$ are two inequivalent irreducible unitary representations of $G$ by assumption, we see that these two identities imply $T_{12} = 0$ and $T_{11} = \lambda I_\mathcal{H}$ for some $\lambda \in \mathbb{C}$ by Schur's lemma (cf. \cite{folland2015}).

Therefore, we have $T = \lambda I_{\mathcal{H} \oplus \mathcal{H}}$, and hence, again by Schur's lemma, we conclude that $\tilde{\pi}$ is an irreducible unitary representation.
\end{proof}

Denote $\overline{p}_0 ^\pm := \mathcal{P} p_0 ^\pm \in X_0 ^\pm$. We will go through the same process as in Sect.~\ref{sec:2} to get a list of irreducible representations associated with the orbit $X_0 ^\pm$, but, this time with a different representative point $\overline{p}_0 ^\pm$.

\begin{proposition}\label{proposition:B.4}
The little group for $\overline{p}_0 ^\pm$ is
\begin{equation}\label{eq:B.11}
\overline{K}:=H_{\overline{p}_0 ^\pm} = \left\{ \begin{pmatrix} z & 0 \\ b & \overline{z} \end{pmatrix} : z \in \mathbb{T} , b \in \mathbb{C} \right\}
\end{equation}
and for each $s \in \frac{1}{2} \mathbb{Z}$, the map $\overline{\eta_s} : \overline{K} \rightarrow U(1)$ defined by
\begin{equation}\label{eq:B.12}
\overline{\eta}_s \begin{pmatrix} z & 0 \\ b & \overline{z} \end{pmatrix} = z^{2s}
\end{equation}
is an irreducible representation of the group $\overline{K}$.
\end{proposition}
\begin{proof}
The proof is exactly the same as in Sect.~\ref{sec:2}
\end{proof}

Following \cite{lee2022b}~Remark~4.4, we form
\begin{equation}\label{eq:B.13}
\rho_{\overline{p}_0 ^\pm , \overline{\eta}_s} (a, \Lambda) := e^{-i \langle  \overline{p}_0 ^\pm  ,  a \rangle} \overline{\eta}_s (\Lambda), \quad (a, \Lambda) \in G_{\overline{p}_0 ^\pm},
\end{equation}
which yields an irreducible unitary representation of $G$ given by
\begin{equation}\label{eq:B.14}
\overline{\pi}_{0, s} ^\pm := \pi_{\overline{ p}_0 ^\pm, \overline{\eta}_s} = \textup{Ind}_{G_{\overline{p}_0 ^\pm}} ^G (\rho_{\overline{p}_0 ^\pm, \overline{\eta}_s} ).
\end{equation}

Denote the primitive bundles formed by $\overline{\eta}_s$ as $\overline{\mathcal{E}}_s$, i.e., $\overline{\mathcal{E}}_s := (H \times_{\overline{\eta}_s} \mathbb{C} , \overline{g} , \overline{\Lambda} )$ associated with $(H, \overline{\eta}_s)$ (cf. \cite{lee2022b}~Propoistion~A.5). The Hermitian metric $\overline{g}$ and the $G$-action $\overline{\Lambda}$ are given by
\begin{equation}\label{eq:B.15}
\overline{g} ( [ B, \zeta_1] , [B, \zeta_2] ) = \overline{\zeta_1 } \zeta_2
\end{equation}
for $B \in H$, $\zeta_1, \zeta_2 \in \mathbb{C}$, and
\begin{equation}\label{eq:B.16}
\overline{\Lambda} (a, A)[B, \zeta] = [AB , e^{-i \langle AB \overline{p}_0 ^\pm , a \rangle} \zeta],
\end{equation}
respectively (cf. \cite{lee2022b}~Lemma~A.9). It is an Hermitian $G$-bundle over $H/\overline{K}$. Note that $\mu$ given in Proposition~\ref{proposition:3.1} is also a $G$-invariant measure on $H/\overline{K} \cong X_0 ^\pm$. By \cite{lee2022b}~Lemma~A.10, we have the following proposition, which is an anlogue of Theorem~\ref{theorem:3.2}.

\begin{theorem}\label{theorem:B.5}
The irreducible representation $\overline{\pi}_{0,s} ^\pm$ is equivalent to the induced representation $\overline{u} : G \rightarrow U\big(L^2 (H/\overline{K}, \overline{\mathcal{E}}_{s}; \mu , g)\big)$ defined as, for $(a, A) \in G$ and $\psi \in L^2(H/K , \mathcal{E}_{\eta_s} ; \mu , g)$,
\begin{equation}\label{eq:B.17}
\overline{u} (a,A) \psi = \overline{\Lambda} (a,A) \circ \psi \circ l_A ^{-1}.
\end{equation} 
\end{theorem}

\begin{lemma}\label{lemma:B.6}
There is an Hermitian bundle isomorphism $\tilde{\varepsilon}: \mathcal{E}_s  \rightarrow \overline{\mathcal{E}}_s$ that covers a diffeomorphism $\varepsilon: H/K \rightarrow H/\overline{K}$ given by
\begin{equation}\label{eq:B.18}
\begin{tikzcd}[baseline=(current  bounding  box.center), column sep=1.5em]
    \mathcal{E}_s \arrow ["{[\Lambda, \lambda] \mapsto [\Lambda^{\dagger -1}, \lambda] }"]{rrrr} \arrow{d} & & & & \overline{\mathcal{E}}_s \arrow{d}
\\
   H/K \arrow ["{\Lambda K \mapsto \Lambda^{\dagger -1} \overline{K} }"]{rrrr} & & & & H/\overline{K}
\end{tikzcd}.
\end{equation}
\end{lemma}
\begin{proof}
The two maps are well-defined since the operation $(\cdot)^{\dagger -1}$ preserves the order of products and $\begin{pmatrix} z & b \\ 0 & \overline{z} \end{pmatrix}^{\dagger -1} = \begin{pmatrix} z & 0 \\ -\overline{b} & 0 \end{pmatrix}$. That the diagram Eq.~(\ref{eq:B.18}) commutes is obvious.

So, $\tilde{\varepsilon}$ is indeed a vector bundle isomorphism covering $\varepsilon$. By Eqs.~(\ref{eq:3.3}) and (\ref{eq:B.15}), it is also an Hermitian bundle isomorphism.
\end{proof}

\begin{lemma}\label{lemma:B.7}
For $(a, A) \in G$ and $[B, \zeta] \in \mathcal{E}_s$, the following holds.
\begin{equation}\label{eq:B.19}
\tilde{\varepsilon} \big( \Lambda(a,A) [B, \zeta] \big) = \overline{\Lambda} (\mathcal{P}a , A^{\dagger -1} ) \tilde{\varepsilon} \big([B, \zeta] \big)
\end{equation}
\end{lemma}
\begin{proof}
\begin{align*}
\tilde{\varepsilon} \big( \Lambda(a,A) [B, \zeta] \big) = \tilde{\varepsilon} ( [AB, e^{-i \langle AB p_0 ^\pm , a \rangle} \zeta] ) = [(AB)^{\dagger -1} , e^{-i \langle \mathcal{P} AB p_0 ^\pm , \mathcal{P} a \rangle} \zeta] \\
= [A^{\dagger -1} B^{\dagger -1} , e^{-i \langle (AB)^{\dagger -1} \mathcal{P} p_0 ^\pm , \mathcal{P} a \rangle } \zeta] = \overline{\Lambda} (\mathcal{P}a, A^{\dagger -1} ) [B^{\dagger -1} , \zeta] \\
= \overline{\Lambda} ( \mathcal{P}a, A^{\dagger -1} ) \tilde{\varepsilon} \big([B, \zeta] \big)
\end{align*}
where we used Lemma~\ref{lemma:B.1} and the fact that $\mathcal{P} \in O(1,3)$.
\end{proof}

Via Theorems~\ref{theorem:3.2} and \ref{theorem:B.5}, let's identify the representations $\pi_{0,s} ^\pm$ and $\overline{\pi}_{0,s} ^\pm$ with the formulae Eqs.~(\ref{eq:3.5}) and (\ref{eq:B.17}), respectively. Then,

\begin{theorem}\label{theorem:B.8}
The map $V : L^2 (H/K , \mathcal{E}_s ; \mu , g) \rightarrow L^2 (H/ \overline{K} , \overline{\mathcal{E}}_s ; \mu , \overline{g})$ defined by
\begin{equation}\label{eq:B.20}
V(\psi) := \tilde{\varepsilon} \circ \psi \circ \varepsilon^{-1}, \quad \psi \in L^2 (H/K , \mathcal{E}_s ; \mu , g)
\end{equation}
is a unitary map intertwining the representation $\pi_{0,s} ^\pm \circ \delta $ with $\overline{\pi}_{0, s} ^\pm$.
\end{theorem}
\begin{proof}
$V$ is unitary by Lemma~\ref{lemma:B.6}. Let $\psi \in L^2 (H/K , \mathcal{E}_s ; \mu , g)$ and $(a,A) \in G$. To see that $V$ intertwines the two representations, we observe
\begin{align*}
V\big( \pi_{0,s} ^\pm \circ \delta (a,A) \psi \big) = V \big( \pi_{0,s} ^\pm (\mathcal{P} a, A^{\dagger -1} ) \psi \big) = \tilde{\varepsilon} \circ \big( \Lambda( \mathcal{P} a , A^{\dagger -1} ) \circ \psi \circ l_{A^{\dagger} } \big) \circ \varepsilon^{-1} \\
= \overline{\Lambda} (a, A) \circ \tilde{\varepsilon} \circ \psi \circ \varepsilon^{-1} \circ l_{A^{-1}} = \overline{\pi}_{0,s} ^\pm (a, A) V (\psi)
\end{align*}
where we used Lemma~\ref{lemma:B.7} and the obvious fact $\varepsilon \circ l_{A} = l_{A^{\dagger -1}} \circ \varepsilon$.
\end{proof}

\begin{theorem}\label{theorem:B.9}
\begin{equation}\label{eq:B.21}
\pi_{0,s} ^\pm \circ \delta \cong \pi_{0, -s} ^\pm
\end{equation}
\end{theorem}

\begin{proof}
Observe that for
\begin{equation*}
R := e^{\pi J^2} = \begin{pmatrix} 0 & -1 \\ 1 & 0 \end{pmatrix} \in SL(2, \mathbb{C}),
\end{equation*}
we have
\begin{equation*}
R p_0 ^\pm = \kappa( R) p_0 ^\pm = \overline{p}_0 ^\pm \in X_0 ^\pm.
\end{equation*}

Now, if we observe
\begin{equation*}
\overline{\eta}_s \left( R \begin{pmatrix} z & b \\ 0 & \overline{z} \end{pmatrix} R^{-1} \right) = \overline{\eta}_s \begin{pmatrix} \overline{z} & 0 \\ -b & z \end{pmatrix} = z^{-2s} = \eta_{-s} \begin{pmatrix} z & b \\ 0 & \overline{z} \end{pmatrix},
\end{equation*}
we can apply \cite{lee2022b}~Theorem~4.2.(c) to show that
\begin{equation*}
\overline{\pi}_{0, s} ^\pm \cong \pi_{0, -s} ^\pm
\end{equation*}

Now, the theorem follows from Theorem~\ref{theorem:B.8}.
\end{proof}

\begin{theorem}[The precise statement of Theorem~\ref{theorem:3.5}]\label{theorem:B.10}
Fix $0 \neq s \in \frac{1}{2} \mathbb{Z}$ and identify $\pi_{0,-s} ^\pm \cong \pi_{0,s} ^\pm \circ \delta$ using Theorem~\ref{theorem:B.9}. Then, the representation $\pi_{0,s} ^\pm \oplus \pi_{0, -s} ^\pm$, extended to $\tilde{G}$ as in Eqs.~(\ref{eq:B.9})--(\ref{eq:B.10}) is an irreducible representation of $\tilde{G}$.
\end{theorem}
\begin{proof}
The proof is now obvious by Theorems~\ref{theorem:B.3} and \ref{theorem:B.9}. One only needs to observe that $\pi_{0,s} ^\pm \ncong \pi_{0, -s} ^\pm$ by Theorem~\ref{theorem:2.7}.
\end{proof}

\section{Proof of Theorem~\ref{theorem:5.2}}\label{sec:C}

Let $\tilde{\mathcal{E}}_s$ be the bundle associated with the principal bundle $H \rightarrow H/K$ and the representation $\tilde{\eta}_s : K \rightarrow U(2)$ (cf. \cite{lee2022b}~Proposition~A.2). Then, \cite{lee2022b}~Lemma~A.9 shows that $\tilde{\mathcal{E}}_{s}$ is an Hermitian $G$-bundle, called the \textit{primitive bundle associated with $\tilde{\eta}_s$}, with the metric
\begin{equation}\label{eq:C.1}
\tilde{g}\big( [B, v] , [B, w] ) = v \cdot w
\end{equation}
where $B \in H $ and $v, w \in \mathbb{C}^2$, and the $G$-action
\begin{equation}\label{eq:C.2}
\tilde{\Lambda} (a, A) [B , v] = [AB , e^{-i \langle A B p_0 , a \rangle }  v]
\end{equation}
where $A, B \in H$, $a \in \mathbb{R}^4$, and $v \in \mathbb{C}^2$.

Eq.~(\ref{eq:4.3}) and \cite{lee2022b}~Lemma~A.10 yield the following theorem.

\begin{theorem}\label{theorem:C.1}
The representation $\Pi_s$ of $G$, which represents spin-$s$ massless particles with parity inversion symmetry (cf. Definition~\ref{definition:3.6}), is equivalent to the induced representation $\tilde{U} : G \rightarrow U\big(L^2 (H/K, \tilde{\mathcal{E}}_{s}; \mu , \tilde{g})\big)$ defined as, for $(a, A) \in G$ and $\psi \in L^2(H/K , \tilde{\mathcal{E}}_{s} ; \mu , \tilde{g})$,
\begin{equation}\label{eq:C.3}
\tilde{U} (a,A) \psi = \tilde{\Lambda}(a,A) \circ \psi \circ l_A ^{-1}.
\end{equation}
\end{theorem}

Therefore, in view of \cite{lee2022b}~Proposition~5.4, the following theorem gives a proof for Theorem~\ref{theorem:5.2}.

\begin{theorem}\label{theorem:C.2}
The following map is a vector bundle isomorphism between the primitive bundle $\tilde{\mathcal{E}}_s$ and the perception bundle $E$ as defined in Theorem~\ref{theorem:5.2}, via which the Hermitian $G$-structure of the primitive bundle (Eqs.~(\ref{eq:C.1})--(\ref{eq:C.2})) is translated into the Hermitian $G$-structure of $E$ as listed in Table~\ref{tab:2}.
\begin{equation}\label{eq:C.4}
\begin{tikzcd}[baseline=(current  bounding  box.center)]
\tilde{\mathcal{E}}_s \arrow[rrrr, hookrightarrow, "{[A,v] \mapsto \left(p, \Phi(A) \epsilon(v) + D_{A p_0}\right)}" ] \arrow{drr}
& & & & E \arrow{dll} \\
& & X_0 & &
\end{tikzcd}.
\end{equation}
\end{theorem}
\begin{proof}
First, since $\Phi(A) \epsilon(v) \in \Phi(A) R = F_{A p_0}$ for $A \in H$, $v \in \mathbb{C}^2$, we see that at least the expression of the map makes sense.

Observe that for $A \in H$, $p = Ap_0 \in X_0$, $B \in K$, and $v \in \mathbb{C}^2$,
\begin{align*}
\Phi(AB) \epsilon(v) + D_p = \Phi (A)\Phi(B) \epsilon (v) + \Phi(A)W = \Phi(A) \big(\Phi(B) \epsilon(v) + W \big) \\
= \Phi(A) \big(\epsilon( \tilde{\eta}_s (B) v) +W \big) = \Phi(A) \epsilon( \tilde{\eta}_s (B) v)  + D_p
\end{align*}
by the condition Eq.~(\ref{eq:5.1}). This equation implies that the map Eq.~(\ref{eq:C.4}) is a well-defined vector bundle homomorphism (cf. \cite{lee2022b}~Proposition~A.2).

If $\Phi(A) \epsilon(v) + D_p = 0 + D_p$, then by multiplying $\Phi(A)^{-1}$ from the left, we see $\epsilon(v) +W = 0 + W$. Since $W \cap \epsilon(\mathbb{C}^2) = 0$ by assumption, we see that this implies $\epsilon(v)=0$, which in turn implies $v=0$. Hence, the map Eq.~(\ref{eq:C.4}) is injective at each fiber. It is also surjective since $E_p \cong F_p / D_p \cong R /W \cong \epsilon(\mathbb{C}^2)$ by the assumption $R = W \oplus \epsilon(\mathbb{C}^2)$ and hence each fiber $E_p$ is 2-dimensional.

We now show the well-definededness of the Hermitian metric $h$ in Table~\ref{tab:2}. I.e., we claim that the expression
\begin{equation}\label{eq:C.5}
h_p (z+ D_p , w+ D_p) = \langle \Phi(A)^{-1} z + W, \Phi(A)^{-1} w + W \rangle_{R/W}
\end{equation}
for $z, w \in F_p$ and $A \in H$ such that $A p_0 = p$ does not depend on the choice of $A$. Let $B \in K$. Note that since $\tilde{\eta}_s (B)$ is a unitary on $\mathbb{C}^2$ and $\epsilon$ is an isometry, Eq.~(\ref{eq:5.1}) and the assumption $W \perp \epsilon(\mathbb{C}^2)$ imply that
\begin{align*}
\Big\langle \Phi(B) \epsilon(v) + W, \Phi(B) \epsilon(w) + W \Big\rangle_{R/W} = \Big\langle \epsilon( \tilde{\eta}_s (B) v) + W, \epsilon( \tilde{\eta}_s (B) w) + W \Big\rangle_{R/W} \\
 = \big\langle \epsilon( \tilde{\eta}_s (B) v) , \epsilon( \tilde{\eta}_s (B) w ) \big\rangle_{\epsilon(\mathbb{C}^2)} = ( \tilde{\eta}_s (B) v) \cdot ( \tilde{\eta}_s (B) w) 
= v \cdot w = \langle \epsilon(v) , \epsilon(w) \rangle_{\epsilon(\mathbb{C}^2)} \\
= \Big\langle \epsilon(v) +W , \epsilon(w) +W \Big\rangle_{R/W}
\end{align*}
for $v , w \in \mathbb{C}^2$. Thus, we see that $\Phi(B)$ acts as a unitary on $R/W$ and hence Eq.~(\ref{eq:C.5}) is well-defined.

The assumption $W \perp \epsilon(\mathbb{C}^2)$ also implies that, for $v, w \in \mathbb{C}^2$ and $A \in H$, 
\begin{align*}
v \cdot w &= \langle \epsilon(v) , \epsilon(w) \rangle_{\epsilon(\mathbb{C}^2)} = \big\langle \epsilon(v) +W , \epsilon(w) +W \big\rangle_{R/W} \\
&= \Big\langle \Phi(A)^{-1} \Phi(A) \epsilon(v) + W , \Phi(A)^{-1} \Phi(A) \epsilon(w) + W \Big \rangle_{R/W},
\end{align*}
which, together with the well-definededness, proves that the isomorphism Eq.~(\ref{eq:C.4}) indeed translate the metric $\tilde{g}$ of $\tilde{\mathcal{E}}_s$ into the metric $h$ of $E$.

Finally, the expression for the $G$-action $\lambda$ in Table~\ref{tab:2} is easily obtained from Eqs.~(\ref{eq:C.2}) and (\ref{eq:C.4}).
\end{proof}

\section*{Acknowledgments}

First and foremost, I want to express my deepest thanks to God for His unending love and mercy toward me.

Secondly, I would like to express my sincere gratitude to the anonymous reviewers of the paper \cite{lee2022b}. Their invaluable comments and suggestions to that paper have been very helpful in shaping this paper. I would also like to thank Dr. Alexei Deriglazov. He has read through the draft of this paper and given me some suggestions that have greatly helped me in revising the manuscript.

H. Lee was supported by the Basic Science Research Program through the National Research Foundation of Korea (NRF) Grant NRF-2022R1A2C1092320.

\section*{Data availability statement}

Data sharing is not applicable to this article as no new data were created or analyzed in this study.

\section*{References}
\bibliographystyle{plain}
\bibliography{bibliography}
\end{document}